\documentclass[runningheads]{llncs}

\usepackage[utf8]{inputenc}
\usepackage{graphicx}
\usepackage[colorlinks=true,bookmarksopen,bookmarksdepth=2]{hyperref}
\usepackage[OT4]{fontenc}

\usepackage{amsmath, amsfonts, amssymb, stmaryrd}
\usepackage{xcolor, cancel}
\usepackage{refcount, listings, ebproof}

\newcommand \alt{\;\;|\;\;}
\newcommand \subtype{<:}

\newcommand{\ie}{i.e.}
\newcommand{\eg}{e.g.}
\newcommand{\wrt}{w.r.t.~}
\newcommand{\etal}{et~al.}

\newcommand{\call}[2]{{{#1}({#2})}}
\newcommand{\fun}[2]{{{#1} \to {#2}}}
\newcommand{\optional}[1]{{{#1}^?\!}}
\newcommand{\nil}{\bullet}
\newcommand{\cons}[2]{{#1 :: #2}}
\newcommand{\none}{\nil}
\newcommand{\some}[1]{{[#1]}}

\newcommand{\lookup}[2]{{\call {#1} {#2}}}
\newcommand{\update}[3]{{{#1}[{#2} := {#3}]}}
\newcommand{\alloc}[3]{{{#1} \ast [{#2} \mapsto {#3}]}}
\newcommand{\transl}[2]{\underline{ {#2}, {#1} }}
\newcommand{\strip}[1]{\lookup {\mathsf{strip}}{#1}}
\newcommand{\emp}{\mathbf{emp}}
\newcommand{\init}{\mathbf{init}}

\newcommand{\dubleton}{\mathit{dub}}

\newcommand{\term}{t}
\newcommand{\val}{v}
\newcommand{\wnf}{w}
\newcommand{\inert}{i}
\newcommand{\nf}{n}
\newcommand{\neu}{a}

\newcommand{\Context}{C}
\newcommand{\Genericctx}{C}
\newcommand{\Weak}{W}

\newcommand{\Fireball}{F}

\newcommand{\plug}[2]{ {{#1}[{#2}]} }
\newcommand{\contr}[1]{ {\rightharpoonup}_{#1} }
\newcommand{\red}[2]{ {\stackrel{#1}{\to}_{#2}} }
\newcommand{\rred}[2]{ {\stackrel{#1}{\twoheadrightarrow}_{#2}} }
\newcommand{\nred}[2]{ {\stackrel{#1}{\not\to}_{#2}} }
\newcommand{\converts}[2]{ {\stackrel{#1}{=}_{#2}} }

\newcommand{\tvar}[1]{{#1}}
\newcommand{\tlam}[2]{{\lambda{#1}.\,{#2}}}
\newcommand{\tapp}[2]{{{#1}\;{#2}}}

\newcommand{\avar}[1]{\ensuremath{V}(#1)}
\newcommand{\abstr}[3]{[#1, #2, #3]}
\newcommand{\cache}[2]{{{#2}^{#1}}}
\newcommand{\clos}[2]{{[ #1, #2 ]}}

\newcommand{\subst}[3]{\update{#3}{#1}{#2}}
\newcommand{\abstrb}[2]{[#1, #2]}
\newcommand{\fresh}[1]{{{#1}^\ast}}
\newcommand{\fv}[1]{ {\mathit{FV}({#1})} }
\newcommand{\bv}[1]{ {\mathit{BV}({#1})} }
\renewcommand{\emptyset}{\varnothing}
\newcommand{\disjoint}[2]{ {#1 \cap #2 = \emptyset} }
\newcommand{\delimit}[1]{{\langle {#1} \rangle}}

\newcommand{\hole}{\Box}
\newcommand{\frapp}[1]{{\tapp{\hole}{#1}}}
\newcommand{\flclos}[2]{\tapp{\clos{#1}{#2}}{\hole}}
\newcommand{\flapp}[1]{{\tapp{#1}{\hole}}}
\newcommand{\flam}[1]{{\tlam {#1} \hole}}
\newcommand{\fcache}[1]{@[{#1}]}

\newcommand{\econf}[5]{{\langle{#3}, {#4}, {#5}, {#1}, {#2}\rangle_{\cal E}}}
\newcommand{\cconf}[4]{{\langle{#3}, {#4}, {#1}, {#2}\rangle_{\cal C}}}
\newcommand{\sconf}[4]{{\langle{#3}, {#4}, {#1}, {#2}\rangle_{\cal{S}}}}
\newcommand{\mconf}[6]{{\langle{#3}, {#4}, {#5}, {#6}, {#1}, {#2}\rangle_{\cal M}}}

\newcommand{\econfb}[2]{{\langle{#1}, {#2}\rangle_{\cal E}}}
\newcommand{\cconfb}[2]{{\langle{#1}, {#2}\rangle_{\cal C}}}
\newcommand{\sconfb}[2]{{\langle{#1}, {#2}\rangle_{\cal{S}}}}
\newcommand{\mconfb}[4]{{\langle{#1}, {#2}, {#3}, {#4}\rangle_{\cal M}}}

\newcommand{\psit}{{ {\Phi_\textsf{t}} }}
\newcommand{\psiw}{{ {\Phi_\textsf{v}}}}
\newcommand{\psis}{{ {\Phi_\textsf{S}}}}
\newcommand{\psih}{{ {\Phi_\textsf{H}}}}
\newcommand{\psik}{{ {\Phi_\textsf{K}}}}
\newcommand{\Psit}[1]{{\call {\psit} {#1}}}
\newcommand{\Psiw}[1]{{\call {\psiw} {#1}}}
\newcommand{\Psis}[1]{{\call {\psis} {#1}}}
\newcommand{\Psih}[1]{{\call {\psih} {#1}}}
\newcommand{\Psik}[1]{{\call {\psik} {#1}}}

\newcommand{\semsub}[2]{\llbracket {#1} \rrbracket_{#2}}
\newcommand{\semt}[1]{{\semsub {#1} {\mathsf{t}}}}
\newcommand{\semv}[1]{{\semsub {#1} {\mathsf{v}}}}
\newcommand{\sems}[1]{{\semsub {#1} {\mathsf{S}}}}
\newcommand{\semk}[1]{{\semsub {#1} {\mathsf{K}}}}

\begin{document}
\title{Strong Call by Value is Reasonable for Time\thanks
{This research is supported by the National Science
    Centre of Poland, under grant number 2019/33/B/ST6/00289.}}

\renewcommand{\orcidID}[1]{\href{https://orcid.org/#1}{\includegraphics[width=8pt]{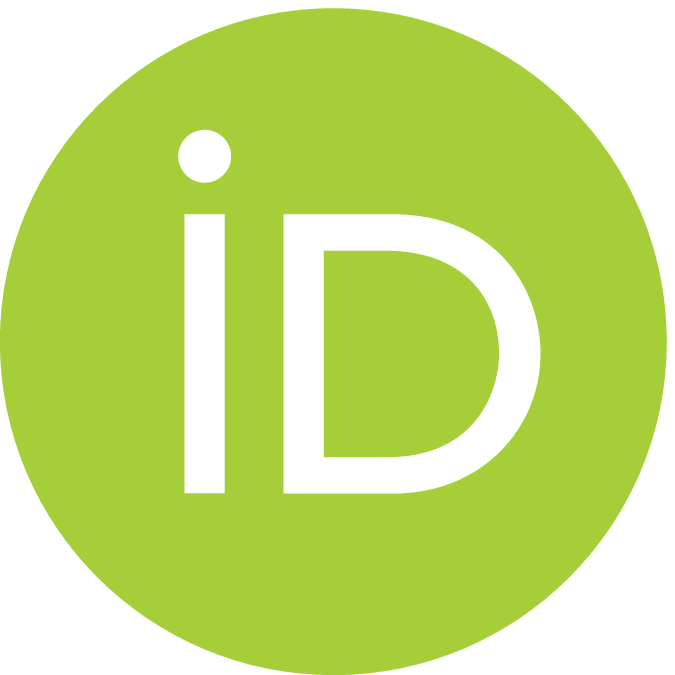}}}

\author{Małgorzata Biernacka\inst{1}\orcidID{0000-0001-8094-0980} \and
Witold Charatonik\inst{1}\orcidID{0000-0001-7062-0385} \and
Tomasz Drab\inst{1}\orcidID{0000-0002-6629-5839}}
\authorrunning{M. Biernacka \etal}

\institute{Institute of Computer Science, University of Wrocław, Poland\\
\email{\{mabi,wch,tdr\}@cs.uni.wroc.pl}\\
{https://ii.uni.wroc.pl/\homedir\{mabi,wch,tdr\}}}

\maketitle

\begin{abstract}
  The invariance thesis of Slot and van Emde Boas states that all
  \emph{reasonable} models of computation simulate each other with
  polynomially bounded overhead in time and constant-factor overhead
  in space. In this paper we show that a family of strong
  call-by-value strategies in the $\lambda$-calculus are reasonable
  for time. The proof is based on a construction of an appropriate
  abstract machine, systematically derived using Danvy \etal's
  functional correspondence that connects higher-order interpreters
  with abstract-machine models by a well-established transformation
  technique. This is the first machine that implements a strong CbV
  strategy and simulates $\beta$-reduction with the overhead
  polynomial in the number of $\beta$-steps and in the size of the
  initial term. We prove this property using a form of amortized cost
  analysis {\`a} la Okasaki.

    \keywords{$\lambda$-calculus \and Abstract machines
    \and Computational complexity \and Reduction strategies
    \and Normalization by evaluation.}
\end{abstract}


\section{Introduction}

\label{sec:intro}
The invariance thesis of Slot and van Emde
Boas~\cite{DBLP:conf/stoc/SlotB84} states that all \emph{reasonable}
models of computation simulate each other with polynomially bounded
overhead in time and constant-factor overhead in space. For a long
time it was not known whether there exist variants of the
$\lambda$-calculus reasonable in this sense. In particular, it was not
known whether there are evaluation strategies for the $\lambda$-calculus
that can be simulated on Turing machines in time bounded by a
polynomial of the number of performed $\beta$-reductions. Recently
this question has been answered positively, for both time and space,
with the weak call-by-value (CbV)~\cite{DBLP:journals/tcs/LagoM08} and,
for time, with the strong call-by-name
(CbN)~\cite{Accattoli-DalLago:LMCS16} strategies. Here, by
constructing an appropriate abstract machine, we show that a family of
strong CbV strategies are reasonable for time.

It is well known that the $\lambda$-calculus provides a foundation for
functional programming languages such as OCaml or Haskell, and more
recently---for proof assistants such as Coq or ELF.  A typical
functional programming language uses a weak reduction strategy (\eg,
call by value or call by need) that only reduces terms until a weak
value is reached, and in particular it does not descend under lambda
abstractions. On the other hand, proof assistants---which can be seen
as functional languages with a rich system of dependent
types---require a strong (\ie, full-reducing) reduction strategy for
type checking.  The study of strong normalization strategies in the
$\lambda$-calculus, their properties and efficient implementation
models, is therefore directly motivated by the need for efficient
tools to handle large-scale, complex verification tasks carried out in
such proof
assistants~\cite{Gregoire-Leroy:ICFP02,Balabonski-al:ICFP17}.

Abstract machines provide implementation models for $\beta$-reduction
that operationalize its key aspects: the process of finding a redex
(decomposition) and the meta-operation of substitution
($\beta$-contraction). They are low-level devices, pretty close to
Turing machines. An abstract machine is typically deterministic and it
makes specific choices when searching for redices, \ie, it implements
a specific reduction strategy of the calculus. Accattoli et al. in his
work on the complexity of abstract
machines~\cite{Accattoli-Coen:LICS15,Accattoli-Barras:PPDP17,DBLP:journals/scp/AccattoliG19}
advocate the following notion of reasonability: a machine is called
\emph{reasonable} if it simulates the strategy with the time overhead
bounded by a polynomial in the number of $\beta$-steps and in the size
of the initial term.

The main technical goal of our work is to construct a reasonable (in
the formal sense defined above) abstract machine for strong CbV. This
has been a nontrivial
\cite{DBLP:conf/rta/Accattoli19,DBLP:conf/ppdp/AccattoliCGC19} open
problem
(partial results include
\cite{AccattoliG16,DBLP:journals/scp/AccattoliG19,BiernackaBCD20}) and
an important step in the research program aimed at developing the
complexity analysis of the Coq main abstract machine formulated by
Accattoli~\cite{Accattoli-Barras:PPDP17}. Since there are known
simulations of Turing machines in the lambda calculus and of abstract
machines by Turing machines, the existence of such a machine implies
reasonability (for time) of strong CbV in the sense of Slot
and van Emde Boas.

\paragraph{Related work.}
In the case of weak reduction there is a vast body of work on abstract
machines, their efficiency and construction techniques, but in the
strong case the territory is largely uncharted. The first machine for
strong normalization in the $\lambda$-calculus is due to Crégut and it
implements Strong CbN~\cite{Cregut:HOSC07}.  A normalization function
realizing strong CbV was proposed by Gr\'egoire\&Leroy and implemented
in their virtual machine extending the ZAM
machine~\cite{Gregoire-Leroy:ICFP02}. Another virtual machine for
strong CbV was derived by Ager et al.~\cite{Ager-al:RS-03-14} from
Aehlig and Joachimski's normalization
function~\cite{Aehlig-Joachimski:MSCS04}. Recently, a strong
call-by-need strategy has been proposed by Kesner et
al.~\cite{Balabonski-al:ICFP17}, and the corresponding abstract
machine has been derived by Biernacka et al.~\cite{BiernackaC19}.
Finally, in a recent work~\cite{BiernackaBCD20}, Biernacka et
al. introduced the first abstract machine for strong CbV, and more
specifically for the right-to-left strong CbV strategy. That machine
has been derived by a series of systematic transformation steps from a
standard normalization-by-evaluation function for the CbV lambda
calculus. All three interpreters of strong
CbV~\cite{Ager-al:RS-03-14,BiernackaBCD20,Gregoire-Leroy:ICFP02} are
not reasonable in the sense defined above.

Our work builds on previous developments in the derivational approach
to the construction and study of semantic artefacts, and in particular
we use Danvy et al's functional correspondence and standard techniques
used in functional programming~\cite{Ager-al:PPDP03}. The
outcome of the derivation is an abstract machine whose control stacks
arise as defunctionalized continuations of the CPS-transformed
evaluator. In the case of strong CbV (as well as other hybrid
strategies) these stacks are not uniform; their structure
corresponds directly to multiple-kinded reduction contexts and this
correspondence can be described by a suitable shape invariant of the
stacks. Hybrid strategies and their connection with machines
have been studied by Garcia-Perez \& Nogueira, and by Biernacka et
al.~\cite{Garcia-Nogueira:SCP14,BChZ-fscd17}.

In between the weak and the strong CbV strategies, one can consider
the Open CbV strategy that extends the usual weak CbV strategy in that
it works on open terms and generalizes the notion of value
accordingly. In a series of articles, Accattoli et al. study Open CbV
in the form of so-called {\em fireball calculus}. They also show
several abstract machines (the GLAM family) and study their complexity
properties with the aim of providing reasonable and efficient
implementations of Open CbV~\cite{DBLP:journals/scp/AccattoliG19}.

\paragraph{Contributions.}
In order to construct a reasonable abstract machine for strong CbV we refine
the techniques from~\cite{BiernackaBCD20}.  Our starting point is the
same NbE evaluator as in~\cite{BiernackaBCD20} but we modify it using
the standard memoization technique to store weak values together with
their strong normal forms (if they get computed along the way).  Next,
we apply Danvy et al's functional correspondence~\cite{Ager-al:PPDP03}
to transform this new interpreter into a first-order abstract
machine. This way the obtained machine is a systematically constructed
semantic artefact rather than an ad-hoc modification of an existing
one. The commented code of the development is available
in~\cite{supplement}.

In order to argue about the efficiency of the derived machine,
we apply a form of amortized cost analysis that uses a potential
function akin to Okasaki's approach~\cite{Okasaki:99}.

Thus, the contributions of this paper include:
\begin{enumerate}
\item a derivation of a reasonable abstract machine for strong CbV,
\item two variants of a reasonable machine: an environment-based one
  and a substitution-based one that uses {\em delimited} substitution,
\item a proof of correctness of the machine and of its reasonability,
\item a corollary that strong CbV is reasonable for time.
\end{enumerate}

\paragraph{Outline.}
In Section~\ref{sec:pre} we introduce the basic concepts of the lambda
calculus and the strong CbV strategy. In Section~\ref{sec:machine} we
present the derivation of the machine starting from a NbE evaluator of
the CbV lambda calculus. In Section~\ref{ssec:eam} we present the
resulting abstract machine with environments and in
Section~\ref{ssec:sam} its substitution-based variant. In
Section~\ref{sec:soundness} we prove the soundness of the machine and
in Section~\ref{sec:complexity} its
reasonability. Section~\ref{sec:conclusion} concludes.

\section{Preliminaries}
\label{sec:pre}

\subsection{\texorpdfstring{Basics of $\lambda$-calculus}{Basics of lambda calculus}}

We work with pure lambda terms given by the following grammar:
\begin{alignat*}{1}
  \term &::= x \alt \tapp {\term_1} {\term_2} \alt \tlam x \term
\end{alignat*}

\noindent where $x$ ranges over some set of identifiers.
As usual, in the sequel the use of a nonterminal restricts the range
of this meta-variable (and its versions with primes or subscripts) to
the defined family.

We define sets of \textit{free} and \textit{bound variables} in a term, and
the \textit{substitution} function as follows
:
\begin{alignat*}{7}
\fv{\tvar x} &= \{ x \} &&&
\bv{\tvar x} &= \emptyset\\
\fv{\tapp {\term_1} {\term_2}} &= \fv{\term_1} &&\cup \fv{\term_2} &
\hspace{12mm}
\bv{\tapp {\term_1} {\term_2}} &= \bv{\term_1} &&\cup \bv{\term_2}\\
\fv{\tlam x {\term}} &= \fv{\term} &&\setminus \{ x \} &
\bv{\tlam x {\term}} &= \bv{\term} &&\cup \{ x \}
\end{alignat*}
\begin{alignat*}{1}
\subst {x} {\term} {\tvar {x'}} &= \begin{cases} \term &: x = x' \\ {\tvar {x'}} &: x \neq x'\end{cases}\\
\subst x {\term} {(\tapp {\term_1} {\term_2}) } &= \tapp {\subst x {\term} {\term_1}} {\subst x {\term} {\term_2}}\\
\subst {x} {\term} {(\tlam {x'} {\term'})} &= \begin{cases} \tlam {x'} {\term'} &: x = x' \\ \tlam {x'} {\subst {x} {\term} {\term'}} &: x \neq x'\end{cases}
\end{alignat*}

This ``raw'' form of substitution makes it possible to capture free
variable occurrences of a substituted term when it is substituted
under lambda (the last case). In order to avoid this problem, it is standard to
introduce $\alpha$-conversion to make sure that bound and free
variables are distinct and the substitution is capture-avoiding.

We first define \text{$\alpha$-contraction} sufficient to rename bound
variables:

\begin{center}\begin{prooftree}
  \hypo{x' \notin \fv \term \cup \bv \term}
\infer1{\tlam {x} \term \; \contr \alpha \; \tlam {x'} {\subst {x} {x'} \term} }
\end{prooftree}\end{center}

In general, a contraction relation is local and can be lifted to a
reduction relation defined on any term by {\em contextual closure} in the
following way.  A \textit{context} is a term with exactly one free
occurrence of a special variable $\hole$ called $\textit{hole}$.
Assuming that $\hole$ is not used as a bound variable, contexts in the
lambda calculus can be defined by the following grammar:
$$\Context ::=
  \tapp \term \Context \alt
  \tapp \Context \term \alt
  \tlam x \Context \alt
  \hole$$

Now, for any contraction relation $\contr{}$ and any context
$\Genericctx$ we define \textit{contextual closure} as follows:
\begin{center}\begin{prooftree}
\hypo{\term_1 \;\contr{} \; \term_2}
\infer1{\plug \Genericctx {\term_1} \;\red{\Genericctx}{}\; \plug \Genericctx {\term_2}}
\end{prooftree}\end{center}

The notation $\plug \Genericctx \term$ is a shortcut for $\subst \hole
\term \Genericctx$ and it denotes a term obtained by {\em plugging}
the hole of the context with the given term.
If the symbol above the arrow is omitted then
the contraction can be done at any location in a
term.

The reflexive-transitive closure of $\red {} {}$ is denoted by $\rred
{}{}$ and the reflexive-symmetric-transitive closure is denoted by
$\converts {}{}$, and is called \textit{conversion}.  Juxtaposition of
two relations denotes their composition, \eg, $s \, \rred {} \beta
\converts {} \alpha\, t$ means that $\exists t'.\; s\,\rred {} \beta
\,t'\, \converts {} \alpha \,t$.

\subsubsection*{Reduction semantics in the lambda calculus.}
Based on the ingredients introduced so far, we can define operational
semantics for the lambda calculus in the form of {\em reduction
  semantics} that specifies $\beta$-contraction as the atomic
computation step, and a {\em reduction strategy} that prescribes
locations in a term where $\beta$-reduction can take place.

We define $\beta$-contraction in the standard way
(and its contextual closure determines $\beta$-reduction):

\begin{center}\begin{prooftree}
  \hypo{\term_1 \; \converts {} \alpha \; \term'_1}
  \hypo{\disjoint {\bv {\term'_1}} {\fv {\term_2}}}
\infer2{\tapp {(\tlam x {\term_1})} {\term_2}
\; \contr \beta \;
\subst x {\term_2} {\term'_1}}
\end{prooftree}\end{center}

In order to define a specific strategy we need to restrict general
contexts $\Genericctx{}{}$. For the CbV strategy
we also need to restrict $\beta$-contraction.

\subsection{Call-by-value strategies}
\label{sec:cbv}
Weak reduction,
\ie, reduction that does not `go under lambda', can be defined as
reduction in weak contexts $\Weak$ defined by the following grammar:
$$\Weak ::= \tapp \term \Weak \alt \tapp \Weak \term \alt \hole$$

In a call-by-value strategy, function arguments need to be evaluated
before the function is applied. If only closed terms are considered,
function arguments are reduced to lambda abstractions,
but in the
open case
we need a more general notion of a~\textit{weak normal form}, which is
a normal form of $\red W \beta$.
Weak normal forms $\wnf$ can be expressed by the following grammar,
where the auxiliary category $\inert$ denotes \textit{inert terms}:
\begin{alignat*}{3}
\wnf &::= & \; \tlam x t &\alt \inert\\
\inert &::= & \tapp \inert \wnf &\alt x
\end{alignat*}
To prevent the substitution of reducible terms we restrict
$\beta$-contraction by requiring that the argument is a weak normal form:
\begin{center}\begin{prooftree}
  \hypo{\term \; \converts {} \alpha \; \term'}
  \hypo{\disjoint {\bv {\term'}} {\fv \wnf}}
\infer2{
\tapp {(\tlam x {\term})} {\wnf}
\; \contr {\beta_\wnf} \;
\subst x {\wnf} {\term'}}
\end{prooftree}\end{center}

This way we obtain the relation $\red \Weak {\beta_\wnf}$, which is
exactly the reduction of the \textit{fireball calculus}, where weak normal
forms are called fireballs~\cite{AccattoliG16}. This calculus is
nondeterministic but strongly confluent and hence all derivations to
weak normal forms have the same length.  It also restores the property
(called {\em harmony} in ~\cite{AccattoliG16}) that its normal
forms are substitutable, \ie, there are no stuck terms.

The fireball reduction can be made deterministic by narrowing the
space of possible evaluation contexts by only allowing to search for
redices in the left part of an application when its right-hand side is a
weak value:
$$F ::= \tapp \term F \alt \tapp F \wnf \alt \hole\\$$

The relation $\red \Fireball {\beta_\wnf}$ (which is
$\beta_\wnf$-contraction in contexts generated by $F$-contexts) is a
restriction of $\red \Weak {\beta_\wnf}$ to a right-to-left strategy,
and it is a~deterministic extension of the closed weak, right-to-left
CbV strategy to the open case.

To obtain a conservative extension of a weak CbV strategy to a
full-reducing one, we can just iterate reduction under lambdas after
reaching a weak normal form.  Such iterations of $\red \Weak
{\beta_\wnf}$ are also strongly confluent and therefore any such
strategy is a good representative of the strong CbV strategy.\footnote{The
applicative order, \ie, leftmost-innermost reduction, where arguments
of a function are directly reduced to a strong normal form, is not a
conservative extension of weak CbV and therefore we do not consider it
as a strong CbV strategy.  Moreover, it does not support recursion
\cite{Sestoft:Jones02} and may take a different number of steps than
strong CbV strategies as defined here.}

In this paper, we study the twice right-to-left\footnote{The
right-to-left direction of evaluation in applications is traditional
and used, \eg, in OCaml.} call-by-value strategy (in short, rrCbV),
that is an example of a deterministic extension of $\red \Fireball
{\beta_\wnf}$ to a fully reducing strategy: it is
deterministic~\cite{BiernackaBCD20}, uses the same contraction
relation, and F-contexts are a subset of the rrCbV context family.
This strategy is denoted $\red R {\beta_\wnf}$ and its grammar of
contexts can be defined using three context nonterminals $H$, $R$ and
$F$ (the latter defined as above for the weak strategy):
\begin{alignat*}{3}
R &::= & \;\tlam x R &\alt H \alt F\\
H &::= & \tapp \inert R &\alt \tapp H \nf
\end{alignat*}
where $\nf$ are normal forms of the strategy $\red {} {\beta}$, and
$\neu$ are neutral terms, both defined as follows:
\begin{alignat*}{3}\nf &::= & \; \tlam x \nf &\alt \neu\\
  \neu &::= & \tapp \neu \nf &\alt x
\end{alignat*}
rrCbV is a hybrid strategy, \ie, it combines different substrategies:
the $F$-contexts define the weak CbV substrategy, $R$ is the starting
nonterminal defining the substrategy that either weakly reduces, or
fully reduces weak normal forms using the auxiliary $H$-strategy,
which in turn is responsible for fully reducing inert terms.


\section{Construction of a Reasonable Machine}
\label{sec:machine}

\subsection{A Higher-Order Evaluator}
Our starting point is the higher-order evaluator
from~\cite{BiernackaBCD20} presented in Figure~\ref{fig:nbe-cbv}.  It
computes $\beta$-normal forms by following the principles of
normalization by evaluation~\cite{Filinski-Rohde:RAIRO05}, where the
idea is to map a $\lambda$-term to an object in the meta-language
(here OCaml) from which a syntactic normal form of the input term can
subsequently be read off.

\begin{figure}[t!!]
  \centering
  \begin{lstlisting}
(* syntax of the lambda-calculus with de Bruijn indices *)
type index = int
type term  = Var of index | Lam of term | App of term * term

(* semantic domain *)
type level = int
type sem = Abs of (sem -> sem) | Neutral of (level -> term)

(* reification of semantic objects into normal forms *)
let rec reify (d : sem) (m : level) : term =
  match d with
  | Abs f ->
    Lam (reify (f (Neutral (fun m' -> Var (m'-m-1))))(m+1))
  | Neutral l ->
    l m

(* sem -> sem as a retract of sem *)
let to_sem (f : sem -> sem) : sem = Abs f

let from_sem (d : sem) : sem -> sem =
  fun d' ->
    match d with
    | Abs f ->
      f d'
    | Neutral l ->
      Neutral (fun m -> let n = reify d' m in App (l m, n))

(* interpretation function *)
let rec eval (t : term) (e : sem list) : sem =
  match t with
  | Var n -> List.nth e n
  | Lam t' -> to_sem (fun d -> eval t' (d :: e))
  | App (t1, t2) -> let d2 = eval t2 e
                    in from_sem (eval t1 e) d2

(* NbE: interpretation followed by reification *)
let nbe (t : term) : term = reify (eval t []) 0
  \end{lstlisting}
  \caption{An OCaml implementation of the  higher-order
    compositional evaluator from~\cite{BiernackaBCD20}: an instance of normalization by
    evaluation for a call-by-value $\beta$-reduction in the
    $\lambda$-calculus.}
  \label{fig:nbe-cbv}
\end{figure}

The evaluator works with lambda terms in de Bruijn notation. In this
notation, lambda terms are generated by the grammar
$\term ::= n \alt \tapp {\term_1} {\term_2} \alt \lambda \term$ where
$n$ ranges over natural numbers. There are two flavours of the
notation: de Bruijn indices, where $n$ denotes the number of lambdas
between the represented variable and its binder; and de Bruijn levels,
where $n$ denotes the number of lambdas between the root of a term and
the binder of the variable. For example,
$\tlam{x} {\tapp{x}{(\tlam{y}{\tapp{y}{x}} )}}$ is represented as
$\lambda\, \tapp {0}{(\lambda\, \tapp {0}{1})}$ with de Bruijn indices
and as $\lambda\, \tapp {0}{(\lambda\, \tapp {1}{0})}$ with
levels. Using de Bruijn notation has a clear benefit of avoiding
problems with $\alpha$-conversion, but there is some cost of having
less intuitive $\beta$-reduction.

The normalization function first evaluates terms into the semantic
domain represented by the recursive type {\tt sem} -- it is completely
standard and implemented by the function {\tt eval}. Then the normal
form is extracted from the semantic object by the function {\tt reify}
that mediates between syntax and semantics in the way known from
Filinski and Rohde's work~\cite{Filinski-Rohde:RAIRO05} on NbE for the
untyped $\lambda$-calculus.

Using Danvy et al.'s functional correspondence~\cite{Ager-al:PPDP03}
between higher-order evaluators and abstract machines, the evaluator
of Figure~\ref{fig:nbe-cbv} is transformed in~\cite{BiernackaBCD20} to
an abstract machine that implements a full-reducing call-by-value
strategy for pure $\lambda$-calculus. The obtained machine is not
reasonable in terms of
complexity~\cite{DBLP:journals/scp/AccattoliG19}: it cannot simulate
$n$ steps of $\beta$-reduction in a number of transitions that is
polynomial in $n$ and in the size of the initial term. The reason is
that it never reuses constructed structures, so it has to introduce
each constructor of the resulting normal form in a separate step. This
can lead to an exponential blow-up, as the following example
shows. Consider a family of terms
$$ \omega := \tlam{x}{\tapp{x}{x}} \hspace{3cm}
e_n := \tlam{x}{\tapp{\tapp{c_n}{\omega}}{x}}$$ where $c_n$ denotes the
$n^\text{th}$ Church numeral.  Each $e_n$ reduces to its normal form in the
number of steps linear in $n$, but the size of this normal form is
exponential in $n$.

\begin{remark}
The interpreters of~\cite{Ager-al:RS-03-14,BiernackaBCD20,Gregoire-Leroy:ICFP02} all suffer from exponential time overhead because of the size explosion problem.
\end{remark}

\subsection{A Pitfall of de Bruijn Indices and Levels}

Consider term families defined as follows (here $x, z$ are free variables):

\begin{alignat*}{7}
& A_0    && := x &\hspace{2cm}
& B_0    && := z \\
& A_{n+} && := \tapp {\tapp {A_n} {(\tlam y {\tvar y})}} x &
& B_{n+} && := \tapp z {\tlam w {B_n}}
\end{alignat*}
$$ Q_n := \tlam x {\tapp {(\tlam z {B_n})} {A_n}}$$
Note that there are $n+1$ free occurrences of variable $z$ in $B_n$,
each of which is under a different number of lambdas (on a different de
Bruijn level).  Terms of family $Q$ are closed and they reach their
normal forms in one \text{$\beta$-reduction} which substitutes $A_n$
for $z$ in $B_n$. The size of $A_n$ is linear in $n$ and it is
substituted for linearly many $z$s in $B_n$ resulting in a normal form
of quadratic size.

Terms $A_n$ can be easily shared in memory and the resulting
representation has size linear in $n$.  It is not however possible
with de Bruijn indices nor levels representation because the resulting
normal form has quadratically many constructors of distinct subterms.
We illustrate this issue with the term $Q_2$ and its normal forms
using 3 different representations: with names, indices and levels.

$$Q_2 = \tlam x {\tapp {(\tlam z {\tapp z {\tlam w {\tapp z {\tlam w z}}}})} {(\tapp {\tapp {\tapp {\tapp x {(\tlam y {\tvar y})}} x} {(\tlam y {\tvar y})}} x})}
$$
\begin{alignat*}{12}
\mathcal N_{\text{nam}}(Q_2) & =
\lambda x.&&(x (\tlam y y)&&x(\tlam y y)&&x)&&
\lambda w.&&(x (\tlam y y)&&x(\tlam y y)&&x)&&
\lambda w.&&(x (\tlam y y)&&x(\tlam y y)&&x)\\
\mathcal N_{\text{ind}}(Q_2) & =
\lambda &&(0\;\,(\lambda 0)&&0\;\,(\lambda 0)&&0)&&
\lambda &&(1\;\,(\lambda 0)&&1\;\,(\lambda 0)&&1)&&
\lambda &&(2\;\,(\lambda 0)&&2\;\,(\lambda 0)&&2)\\
\mathcal N_{\text{lev}}(Q_2) & =
\lambda &&(0\;\,(\lambda 1)&&0\;\,(\lambda 1)&&0)&&
\lambda &&(0\;\,(\lambda 2)&&0\;\,(\lambda 2)&&0)&&
\lambda &&(0\;\,(\lambda 3)&&0\;\,(\lambda 3)&&0)
\end{alignat*}
Here we have three occurrences of the term
$A_2=(x (\tlam y y) x (\tlam y y) x)$, each on a different de Bruijn
level. Therefore, in the index notation, the subterm $x$ has a
different representation in each of these occurrences. Similarly, in
the level notation, the subterm $\tlam y y$ has a different
representation in each of these occurrences. In consequence, none of
the two notations allows sharing of different occurrences of
$A_2$. This example shows that when working with de Bruijn
representations it is not possible to bound the amount of work per
single \text{$\beta$-reduction} by a quantity proportional to the size
of the initial term.  It does not mean that machines working with these
representations must be unreasonable but a proof involving this
quadratic dependency may be more complex than one without this
problem.  Therefore, we choose to work with names.

\subsection{A Reasonable Higher-Order Evaluator}
\label{sec:reasonableEvaluator}
Now we present a higher-order evaluator that will be transformed
into an abstract machine.  The implementation is given in
OCaml~\cite{Leroy-al:OCaml-4.10}.  It is a modified version of the
evaluator from Figure~\ref{fig:nbe-cbv}, with three major
changes. First, it uses names instead of de Bruijn indices to
represent variables in lambda terms. Second, it abstracts the
environments (the second argument of the function \texttt{eval}) in
the sense that they are no longer directly implemented as lists, but
as different data structures implementing dictionaries. Third, it uses
caches as a form of sharing.

\subsubsection{Terms.}
We start with the syntax of $\lambda$-terms with names:
\begin{lstlisting}
type identifier = string
type term = Var of identifier
          | App of term * term
          | Lam of identifier * term
\end{lstlisting}
\subsubsection{Environments.}
Environments are dictionaries storing values assigned to identifiers.
To handle open terms an environment returns abstract variables for undefined identifiers (with
the same name) and we make sure 
they will not be captured
during abstraction
reification.  Here we extend the name in order to mark that the
variable is free and then free variables of the initial term are replaced
with variables with the extended name in the resulting term.
\begin{lstlisting}
module Dict = Map.Make(
  struct type t = identifier let compare = compare end)

type env = sem Dict.t

let rec env_lookup (x : identifier) (e : env) : sem =
  match Dict.find_opt x e with
  | Some v -> v
  | None   -> abstract_variable (x ^ "_free")
\end{lstlisting}
\subsubsection{Caches.} To achieve a reasonable implementation
for strong CbV we need to introduce a form of
sharing in order to avoid the size explosion problem.

We employ a mechanism similar to memothunks. This allows us to reuse
already computed subterms in normal forms.  An $\alpha$-cache is a
place where a result of type $\alpha$ can be stored and later used to
prevent invoking the same delayed computation many times. It is
implemented as follows:
\begin{lstlisting}
type 'a cache = 'a option ref

let cached_call (c : 'a cache) (t : unit -> 'a) : 'a =
  match !c with
  | Some y -> y
  | None   -> let y = t () in c := Some y; y
\end{lstlisting}
\subsubsection{Values.}
In the original strong CbV evaluator there are two kinds of
values: abstractions and delayed neutral terms which after
defunctionalization correspond to weak normal forms and inert terms,
respectively.  Here we add another constructor to allow annotation
of values with caches for their normal forms.
\begin{lstlisting}
type sem = Abs of (sem -> sem)
  | Neutral of (unit -> term)
  | Cache of term cache * sem
\end{lstlisting}
\subsubsection{Reification.}
In normalization by evaluation reification plays a role of a read back
of concrete syntactic objects from abstract semantic values.  It 
corresponds to full normalization of weak normal forms in the obtained
machine.

Here reification uses two auxiliary functions.  Fresh names are
generated in a standard way using one extra memory cell. In
the representation with names, abstract variables are delayed neutral terms
which just return a variable for a given identifier.
\begin{lstlisting}
let gensym : unit -> int =
  let c = ref 0 in
  fun () ->
    let res = !c in
    c := res + 1;
    res

let abstract_variable (x : identifier) : value =
  let vx = Var x in
  Neutral (fun () -> vx)
\end{lstlisting}
The reification of abstractions and neutral terms is accomplished just as
in the original evaluator: abstractions are called with an abstract
variable with a freshly generated name, and the result is reified under
lambda with the same name; delayed neutral terms are simply forced.
In the case of values with cache lookup, \texttt{Cache (c,v)}, if
the result of reification of \texttt{v} is known, it is simply read from
the cache \texttt{c}; otherwise it is computed and stored in the cache.

\begin{lstlisting}
let rec reify : sem -> term =
  function
  | Abs f ->
    let xm = "x_" ^ string_of_int (gensym ()) in
    Lam (xm, reify (f @@ abstract_variable xm))
  | Neutral l ->
    l ()
  | Cache (c, v) -> cached_call c (fun () -> reify v)
\end{lstlisting}

\subsubsection{Value application.}
Values as elements of a $\lambda$-calculus model should also play role
of endofunctions on themselves.  Abstractions are wrappings of such
functions so it is enough to unwrap them.  A neutral term applied to a
normal form creates a new neutral term, and it can be constructed by
delaying the forcing of this neutral term and the reification of the
argument.  However, if a neutral term is annotated with cache, the
cache should be consulted when the delayed computation is forced.  In
contrast, application does not normalize abstractions (abstractions
are changed by reduction) and therefore their caches are ignored in
such a situation.
\begin{lstlisting}
let rec from_sem : sem -> (sem -> sem) =
  function
  | Abs f                -> f
  |           Neutral l  -> apply_neutral l
  | Cache (c, Neutral l) -> apply_neutral
                            (fun () -> cached_call c l)
  | Cache (c,         v) -> from_sem v
and apply_neutral (l : unit -> term) (v : sem) : sem =
  Neutral (fun () -> let n = reify v in App (l (), n))
\end{lstlisting}
\subsubsection{Evaluation.} Evaluation uses an auxiliary function that
annotates a value with the empty cache, provided it is not already
annotated.  Evaluation reads a $\lambda$-expression as source code.
Source variables merely indicate values in environment.  Source
abstractions are translated to abstraction values that evaluate their
bodies with environments extended by an argument annotated with cache.
Applications are evaluated right-to-left and the left value is applied to
the right one as described earlier.
\begin{lstlisting}
let mount_cache (v:sem) : sem =
  match v with
  | Cache(_,_) -> v
  | _          -> Cache(ref None, v)

let rec eval (t : term) (e : env) : sem =
  match t with
  | Var x        -> env_lookup x e
  | Lam (x, t')  -> to_sem
        (fun v -> eval t' @@ Dict.add x (mount_cache v) e)
  | App (t1, t2) -> let v2 = eval t2 e
                    in from_sem (eval t1 e) v2
\end{lstlisting}
The normalization-by-evaluation function is the composition of
evaluation in the empty environment and reification.
\begin{lstlisting}
let nbe (t : term) : term = reify (eval t Dict.empty)
\end{lstlisting}


\subsection{Abstract Machine}
\label{ssec:eam}
Using Danvy et al.'s functional correspondence~\cite{Ager-al:PPDP03},
the evaluator constructed in Section~\ref{sec:reasonableEvaluator} is
transformed to an abstract machine. The most important steps in this
transformation are the same as on the path from the evaluator in
Figure~\ref{fig:nbe-cbv} to its corresponding abstract machine
in~\cite{BiernackaBCD20}: closure conversion, transformation to
continuation passing style, defunctionalization of continuations to
stacks, entanglement of defunctionalized form to an abstract
machine.
All
these transformations are described in the supplementary
materials~\cite{supplement}.  The machine obtained by derivation is
presented in Figures~\ref{fig:syntax} and~\ref{fig:transitions}.
\begin{figure}[tb]
\begin{alignat*}{3}
\mathit{Identifiers} \ni && x\\
\mathit{Terms}  \ni && \term   &::= \tvar{x} \alt \tapp{\term_1}{\term_2} \alt \tlam x \term\\
\mathit{Locations} \ni && \ell\\
\mathit{Values} \ni && \val   &::= \avar{x} \alt \tapp{\val_1}{\val_2} \alt \abstr x \term E \alt \cache \ell \val\\
\mathit{Envs}   \ni && E   & \;\subtype \fun {\mathit{Identifiers}} {\mathit{Values}} \\
\mathit{Frames} \ni && F   &::= \flclos \term E
    \alt \frapp \val
    \alt \flapp \val
    \alt \frapp \term
    \alt \flam x
    \alt \fcache \ell \\
\mathit{Stacks} \ni && S   &::= \nil \alt \cons{F}{S} \\
\mathit{Term\;Optionals} \ni && \optional \term & ::= \none \alt \some \term\\
\mathit{Heaps} \ni && \;H & \;\subtype \fun {\mathit{Locations}} {\mathit{Term\;Optionals}}\\
\mathit{Counters} \ni && m & \;\in\; \mathbb{N}\\
\mathit{Confs} \ni && K   &::=
  \econf mH \term E S
  \alt \cconf mH S \val 
  \alt \sconf mH S \term\\
  &&&\alt \mconf mH {\optional \term} \ell S \val
\end{alignat*}
\caption{Syntactic categories used in the environment-based machine}
\label{fig:syntax}
\end{figure}

\subsubsection{Values.}  Machine representations of values are
representations of weak normal forms that can additionally be
annotated with heap locations. These locations are used to cache full
normal forms. The grammar presented here is a bit counter-intuitive as
it allows nesting of locations like $(v^\ell)^{\ell'}$ or applications
of closures to other values. The machine maintains invariants
guaranteeing that there are no nested locations and that all values
are decoded to weak normal forms (in particular, all applications
involve inert terms). It is possible to write a more precise grammar
here; however, this would lead to many new syntactic categories and to an
increase in the number of transitions of the resulting machine.

\subsubsection{Environments.}
Environments are dictionaries whose keys are identifiers and values
are annotated machine values of the form $\cache \ell \val$.
They represent assignments of values (weak normal forms) to
variables. They can be implemented as association lists (as it is done
explicitly in~\cite{BiernackaBCD20}); other, more efficient options
are considered in Section~\ref{ssec:transition-costs}.

The content of the initial environment $\init$ corresponds to the
result of lookup in the higher-order evaluator when nothing is assigned to
any variable.  We represent it by the empty collection, but as
stated before, the initial environment in our implementation returns an
abstract variable $\avar {x_{\mathit{free}}}$ for variable $x$.

\subsubsection{Stack and Frames.} Stacks are machine representations
of contexts; technically they are just sequences of frames. The first
five frames in the grammar of $F$ are the same as in the KNV machine of~\cite{BiernackaBCD20}.
They are used in representations of rrCbV
contexts, maintaining the same invariants as KNV;  we discuss it in
Section~\ref{ssec:shape}. The last frame, $\fcache \ell$, has similar
meaning to an extra frame in Crégut's KL and KNL of
\cite{Cregut:HOSC07}: it is used to cache a computed normal form under
location $\ell$. Here we use heap to indicate explicitly which
structures of the machine have to be mutable in an implementation.

The stack is the only mechanism responsible for managing the continuation;
in particular the machine has no component that could be recognized as
a \textit{dump} which is present for example in
\cite{DBLP:journals/scp/AccattoliG19}.

\subsubsection{Counter.}
The machine has a counter that is stored in every configuration and is
not duplicated.  It can be seen as a register.  Its role is to
generate fresh names for abstract variables (see
transition~(\ref{tre:9}) in Figure~\ref{fig:transitions}) and it is only
incremented.  It could also be decremented in rule~(\ref{tre:18}) to
maintain de Bruijn level as it is done in KN and KNV, but then
the freshness of variables would be less obvious.

\subsubsection{Heap.} The heap can be seen as a dictionary whose keys are
locations and (optional) values are terms in normal form.  As it is
the case with the counter, the heap appears in every configuration
exactly once and hence it can be implemented in the RAM model as mutable
memory. Then locations $\ell$ can be seen as pointers to such memory
locations.  In the purely functional setting this can be simulated with a
dictionary, causing logarithmic overhead.  However, even with the use of
mutable state, a configuration of the machine can be seen as a
\textit{persistent data structure}. This is because every mutable
location is used only to memoize the normal form of a determined
value; when the normal form is computed, the stored value never
changes.
No mutable pointers are stored in the heap.  Therefore one can say
that every pointer points to a lower level in the pointing hierarchy.
Thus no \textit{reference cycles} are created and garbage collection
can rely solely on \textit{reference counting}.

In terms of \cite{Accattoli-Barras:PPDP17} the heap can be recognized
as a \textit{global environment}.  It is introduced in the evaluator
and preserved by the derivation, so it is present in the machine.
Moreover, together with the local environment used in the machine, it can
be seen as a derived \textit{split environment}, because every value in
the local environment has to be annotated with a location coming from a
distinguished set of identifiers pointing to the global environment.

\subsubsection{Configurations.} The machine uses four kinds of
configurations corresponding to four modes of operation: in
$\cal{E}$-configurations
the
machine evaluates some subterm to a weak normal form; in $\cal{C}$-configurations
it continues with a computed
weak normal form and in
$\cal{S}$-configurations it
continues with a computed strong normal form. $\cal{M}$-configurations
are used to manipulate access to 
memory.

\subsubsection{Transitions.} The transitions of the machine are presented
in Figure~\ref{fig:transitions}. The first one loads an input term to
the initial configuration; similarly the last one unloads the computed
strong normal form from the final configuration.

Transitions~(\ref{tre:1})--(\ref{tre:3}) are completely standard. In
order to evaluate an application of terms, transition (\ref{tre:1})
calls the evaluation of the argument and pushes the current context
represented by a closure pairing the calling function with the current
environment to the stack. Note that this implements the right-to-left
choice of the order of evaluation of arguments. A~lambda abstraction
in (\ref{tre:2}) is already a weak normal form, so we simply change
the mode of operation to a $\cal{C}$-configuration.
Transition (\ref{tre:3})
simply reads a value of a variable from the environment (which always
returns a wnf) and changes the mode of operation.

\begin{figure}[tb]
\begin{align}
   \term &\mapsto \econf 0 \emp \term \init \nil \nonumber\\
\econf mH {\tapp{\term_1}{\term_2}} E {S_1} &\to \econf mH {\term_2} E {\cons{\flclos{\term_1} E} {S_1}} \label{tre:1}\\
\econf mH {\tlam x \term} E {S_1} &\to \cconf mH {S_1} {\abstr x \term E}\label{tre:2}\\
\econf mH {\tvar x} E {S_1} &\to \cconf mH {S_1} {\lookup E x}\label{tre:3}\\
\cconf mH {\cons {\flclos \term E} {S_1}} \val &\to \econf mH \term E {\cons{ \frapp \val} {S_1}}\label{tre:4}\\
\cconf mH {\cons {\frapp {\cache \ell \val}\!\!} {\!S_1}} {\cache {\phantom \ell} {\abstr x \term E}} &\to \econf mH \term {\update E x {\cache \ell \val}} {S_1}\label{tre:5}\\
\cconf mH {\cons {\frapp  {\cache {\xcancel \ell} \val}\!\!} {\!S_1}} {\cache {\phantom \ell}  {\abstr x \term E}} &\to \cconf m {\alloc H \ell \nil} {\cons {\frapp {\cache \ell \val}} {S_1}} {\abstr x \term E}\label{tre:6}\\
\cconf mH {\cons {\frapp  {\cache {\phantom \ell} \val}\!\!} {\!S_1}} {\cache \ell  {\abstr x \term E}} &\to \cconf mH {\cons {\frapp {\cache {\phantom \ell} \val}} {S_1}} {\abstr x \term E}\label{tre:7}\\
\cconf mH {\cons {\frapp {\val}} {S_1}} {\inert} &\to \cconf mH {S_1} {\tapp {\inert} {\val}}\label{tre:8}\\
\cconf mH {S_3} {\abstr x \term E} &\to \nonumber \\ &\econf {m+1} {\alloc {H} {\ell} \nil} \term {\update E x {\cache \ell {\avar {x_m}}}} {\cons {\flam {x_m}} {S_3}}\label{tre:9}\\
\cconf mH {S_2} {\avar x} &\to \sconf mH {S_2} {\tvar x}\label{tre:10}\\
\cconf mH {S_2} {\tapp {\inert} {\val}} &\to \cconf mH {\cons {\flapp {\inert}} {S_2}} {\val}\label{tre:11}\\
\cconf mH {S_2} {\cache \ell \val} &\to \mconf mH {\lookup H \ell} \ell {S_2} \val\label{tre:12}\\
\mconf mH {\some \nf} \ell {S_2} \val &\to \sconf mH {S_2} \nf\label{tre:13}\\
\mconf mH \none \ell {S_2} \val &\to \cconf mH {\cons {\fcache \ell} {S_2}} \val\label{tre:14}\\
\sconf mH {\cons {\fcache \ell} {S_2}} \nf &\to \sconf m {\update H \ell {\some \nf}} {S_2} \nf\label{tre:15}\\
\sconf mH {\cons {\flapp \val} {S_2}} \nf &\to \cconf mH {\cons {\frapp \nf} {S_2}} \val\label{tre:16}\\
\sconf mH {\cons{\frapp {\nf}} {S_2}} {\neu} &\to \sconf mH {S_2} {\tapp {\neu} {\nf}} \label{tre:17}\\
\sconf mH {\cons {\flam x} {S_3}} \nf  &\to \sconf mH {S_3} {\tlam x \nf}\label{tre:18}\\
  \setcounter{equation}{-2}
\sconf mH \nil \term &\mapsto \term \nonumber
\end{align}

\caption{Transition rules of the environment-based
  machine}\label{fig:transitions}
\end{figure}

Configurations of the form $\cconf mH S \val$ continue with a wnf
$\val$ in a~context represented by stack $S$, with heap $H$ and the
value of the variable counter equal to $m$. There are two goals in
these configurations: the first is to finish the evaluation (to wnfs)
of the closures stored on the stack $S$ according to the weak
call-by-value strategy; the second is to reduce $\val$ to a strong
normal form. This is handled by rules (\ref{tre:4})--(\ref{tre:12}),
where rules (\ref{tre:4})--(\ref{tre:8}) are responsible for the first
goal, and rules (\ref{tre:9})--(\ref{tre:12}) for the second.  The
choice of the executed transition is done by pattern-matching on
$\val$ and $S$; we always choose the first matching transition. This
means, in particular, that $\frapp {\cache \ell \val}$ is a right
application of an annotated value in transition (\ref{tre:5}), but in
transition~(\ref{tre:6}) we have a right
application of a value that is not annotated.

In rule (\ref{tre:4}) the stack contains a~closure, so we start
evaluating this closure and push the already computed wnf to the
stack; when this evaluation reaches a~wnf, rules (\ref{tre:6}) or
(\ref{tre:7}) applies. Rule (\ref{tre:6}) is responsible for an
application of a not-annotated abstraction closure to a not-annotated
wnf; in this case a new location $\ell$ is created on the heap and the
wnf is annotated with $\ell$ (later the computed normal form of the
wnf may be stored in $\ell$).  Rule (\ref{tre:5}) is responsible for
an application of a not-annotated lambda abstraction to an annotated
wnf and implements $\beta$-contraction by evaluating the body of the
abstraction in the appropriately extended environment. Rule
(\ref{tre:7}) is responsible for an application of an annotated
abstraction; since $\beta$-contraction is going to change the normal
form of the computed term, we simply remove the information about the
cache; this rule is immediately followed by transition (\ref{tre:5})
or (\ref{tre:6}) and (\ref{tre:5}).  Rule~(\ref{tre:8}) applies when
$\inert$ is an inert term and $\val$ is an arbitrary (annotated or
not) wnf; in this case we reconstruct the application of this inert
term to the wnf popped from the stack (which gives another
wnf).

Rules (\ref{tre:9})--(\ref{tre:12}) are applied when there are no more
wnfs on the top of the stack; here we pattern-match on the currently
processed wnf $\val$. If it is a closure, transition (\ref{tre:9})
implements the `going under lambda' step: it pushes the elementary
context $ {\flam {x_m}}$ to the stack (note that $x_m$ is a fresh
variable here), increments the variable counter, creates a new
location $\ell$ on the heap, adds the (annotated) abstract variable
${\cache \ell {\avar {x_m}}}$ to the environment, and starts the
evaluation of the body.  If $\val$ is an abstract variable, we reach
a~normal form; rule (\ref{tre:10}) changes the mode of operation to
a $\cal{S}$-configuration.
If
$\val$ is an application $\tapp {\inert} {\val}$, rule (\ref{tre:11})
delays the normalization of $\inert$ by pushing it to the stack and
continues with $\val$; note that this implements the second of our
right-to-left choices of the order of reduction. Finally, if $\val$ is
an annotated value, we change the mode of operation to an $\cal{M}$-configuration
to check if the
normal form has already been computed.

In configuration $\mconf {m}{H}{\optional \term}{\ell}{S}{\val}$  heap $H$
contains location $\ell$ pointing to $\optional \term$ which is an optional of
$\val$'s strong normal form. If it contains the normal form $\nf$ then
transition~(\ref{tre:13}) immediately returns $\nf$ and changes the mode
of operation to an $\cal{S}$-configuration.
Otherwise, if it
is empty, rule~(\ref{tre:14}) pushes the location $\ell$ to the
stack and calls the normalization of $\val$. When this normalization
is finished, rule~(\ref{tre:15}) stores the computed normal form at
location $\ell$.

Configurations of the form $\sconf mH St$ continue with a~(strong)
normal form~$t$ in a~context represented by $S$. The goal in these
configurations is to finish the evaluation of inert terms stored on
the stack and to reconstruct the final term. This is handled by
transitions (\ref{tre:15})--(\ref{tre:18}); the choice of the
transition is done by pattern-matching on the stack. If there is an
inert term $\val$ on the top of the stack, rule (\ref{tre:16}) pushes
the already computed normal form on the stack and calls normalization
of $\val$ by switching the mode of operation to
a $\cal{C}$-configuration.
Otherwise there is
a~previously computed normal form or a~$\flam {x}$ frame on the top of
the stack; in these cases transitions (\ref{tre:17}) and (\ref{tr:18})
reconstruct the term accordingly. Finally, when the stack is empty,
the machine stops and unloads the final result from a~configuration.

\begin{remark}
  The machine works on a term representation that allows sharing, as
  in OCaml or in \cite{CondoluciAC19}. In particular, caching of the
  computed normal forms allows their reuse in the construction of the
  output term. For example, exponentially big normal forms of the
  family $e_n = \tlam{x}{\tapp{\tapp{c_n}{\omega}}{x}}$ consume only a
  linear in $n$ amount of memory and are computed in linear time. We
  assume that unloading of the final result does not involve unfolding
  it to the unshared term representation (which might require exponential
  time).
\end{remark}


\section{A Substitution-Based Machine}

The machine from Section~\ref{ssec:eam} uses environments to represent
delayed substitutions. In this section we replace environments with
\emph{delimited substitutions}.
This variant provides an intermediate step between
machine configurations and their decodings, and makes correctness proofs in
Section~\ref{sec:soundness} easier to follow.

\subsection{From variables as pointers to delimited substitution}
Accattoli and Barras in \cite{Accattoli-Barras:PPDP17} present a
technique that represents \emph{variables as pointers} and enables
substitution and lookup of a substituted value in constant time. They
use it to obtain an efficient version of Milner Abstract Machine.
Here we apply this technique to eliminate local environments.

The technique relies on a modified representation of terms, where
identifiers in abstractions and variables are replaced by pointers to
mutable memory. With this representation, a substitution can be
performed in constant time by simply writing to an appropriate memory
cell; similarly the lookup of a variable in an environment boils down to
a single reading operation. However, by overwriting a variable, the
substitution destroys the original term, which, in consequence, cannot
be shared. In order to avoid incorrect modification of shared terms,
some subterms must be copied. In~\cite{DBLP:journals/scp/AccattoliG19}
it is observed that a new copy is needed only when it comes to
$\beta$-reduction: only bodies of $\lambda$-abstractions must be
copied and the only identifier that needs to point to a fresh memory
cell is the one of the abstraction to be applied. Intuitively, this
corresponds to an $\alpha$-renaming of a $\lambda$-abstraction before
a $\beta$-reduction step. It is called \textit{renaming on $\beta$}.

In the next section we fuse copying a term and overwrite-based
substitution into a substitution function working on purely functional
term representation.  The tricky bit is that the copying function does not
copy already substituted variables but shares them instead. To
preserve this property, we use \emph{substitution delimiters} of the form
$\delimit \val$ informing that there are no occurrences of the
substituted variable in subterm~$\val$ and that $\val$ can be shared.

\subsection{A substitution-based machine}
\label{ssec:sam}
The new machine is presented in Figures~\ref{fig:notationsSKNV}
and~\ref{sKNCV}. The transitions in Figures~\ref{fig:transitions}
and~\ref{sKNCV} are in one-to-one correspondence and the two machines
bisimulate each other.
The style of the presentation of the new machine is more implicit,
\eg, the mechanism of fresh name generation and the heap component, though
implicitly present, are not spelled out in configurations. Some new
notation is used in Figure~\ref{sKNCV}: $\fresh{x}$ denotes a fresh
variable, $\lookup H \ell$ is the content of location $\ell$ on the
implicit heap $H$ and $\update {} \ell \cdot$ is an update of location $\ell$.

The syntax of terms is extended with a substitution delimiter
$\delimit{}$ carrying substituted values.  It is the only change in
term representation; input terms do not have to be compiled before
being loaded to the machine because they fall within the extended grammar.
Output terms are free of substitution delimiters.

Every configuration of the environment-based machine can be translated
into a substitution-based one by executing delayed substitution for
free variables and denoting it with substitution delimiters. The main
part of this translation, \ie, translation of closures
$\transl \cdot \cdot$, is given below, where
${E \setminus [x \mapsto -]}$ denotes the environment 
$E$ with removed binding for $x$. This establishes a strong
bisimulation between the machines.

\begin{alignat*}{1}
\transl E {\tvar {x}} &= \begin{cases}\delimit {\lookup E x} &: x \in E\\
x &: x \notin E\end{cases}\\
\transl E {\tapp {\term_1} {\term_2} } &= \tapp {\transl E {\term_1}} {\transl E {\term_2}}\\
\transl E {\tlam {x} {\term}} &= \tlam {x} {\transl {E \setminus [x \mapsto -]} \term}
\end{alignat*}

\begin{figure}[htb]
\textbf{Syntax:}
\begin{alignat*}{3}
x &\;\;\in \mathit{Identifiers}\\
\ell & \;\;\in \mathit{Locations}\\
\term   &::= \tvar{x} \alt \tapp{\term_1}{\term_2} \alt \tlam x \term \alt \delimit \val \\
\val   &::= \avar{x} \alt \tapp{\val_1}{\val_2} \alt \abstrb x \term \alt \cache \ell \val\\
F   &::= \flapp \term
    \alt \frapp \val
    \alt \flapp \val
    \alt \frapp \term
    \alt \flam x
    \alt \fcache \ell \\
S   &::= \nil \alt \cons{F}{S} \\
\optional \term & ::= \none \alt \some \term\\
K   &::=
  \econfb \term S
  \alt \cconfb S \val 
  \alt \sconfb S \term
  \alt \mconfb {\optional \term} \ell S \val
\end{alignat*}
\textbf{Substitutions:}
\begin{equation*}
\begin{aligned}[c]
\subst {x} {\term} {\tvar {x'}} &= \begin{cases} \term &: x = x' \\ {\tvar {x'}} &: x \neq x_1\end{cases}\\
\subst x {\term} {\tapp {\term_1} {\term_2} } &= \tapp {\subst x {\term} {\term_1}} {\subst x {\term} {\term_2}}\\
\subst {x} {\term} {(\tlam {x'} {\term'})} &= \begin{cases} \tlam {x'} {\term'} &: x = x' \\ \tlam {x'} {\subst {x} {\term} {\term'}} &: x \neq x'\end{cases}\hspace{7mm}\\
\subst x {\term} {\delimit \val} &= \delimit \val\\
\strip{x}                 &= \avar x\\
\strip {\delimit \val}   &= \val
\end{aligned}
\end{equation*}
\caption{Syntactic categories and substitutions used in the substitution-based machine}
\label{fig:notationsSKNV}
\end{figure}

\begin{figure}[htb]
\begin{align}
\setcounter{equation}{0}
   \term &\mapsto \econfb \term \nil \nonumber\\
\econfb {\tapp{\term_1}{\term_2}} {S} &\to \econfb {\term_2} {\cons {\flapp{\term_1}} {S}}  \label{tr:1}\\
\econfb {\tlam x \term} {S} &\to \cconfb {S} {\abstrb x \term}\label{tr:2}\\
\econfb {\term} {S} &\to \cconfb {S} {\strip \term}\label{tr:3}\\
\cconfb {\cons {\flapp \term} {S}} \val &\to \econfb \term {\cons{ \frapp \val} {S}}\label{tr:4}\\
\cconfb {\cons {\frapp {\cache \ell \val}\!\!} {\!S}} {\cache {\phantom \ell} {\abstrb x \term}} &\to \econfb {\subst x {\delimit {\cache \ell \val}} \term} {S}\label{tr:5}\\
\cconfb {\cons {\frapp  {\cache {\phantom \ell} \val}\!\!} {\!S}} {\cache {\phantom \ell}  {\abstrb x \term}} &\to \cconfb {\cons {\frapp {\cache \ell \val}} {S}} {\abstrb x \term}\label{tr:6}\\
\cconfb {\cons {\frapp  {\cache {\phantom \ell} \val}\!\!} {\!S}} {\cache \ell  {\abstrb x \term}} &\to \cconfb {\cons {\frapp  {\cache {\phantom \ell} \val}} {S}} {\abstrb x \term}\label{tr:7}\\
\cconfb {\cons {\frapp {\val_2}} {S}} {\val_1} &\to \cconfb {S} {\tapp {\val_1} {\val_2}}\label{tr:8}\\
\cconfb {S} {\abstrb x \term} &\to \econfb {\subst x {\delimit {\cache \ell {\avar {\fresh x}}}} \term} {\cons {\flam {\fresh x}} {S}} \label{tr:9}\\
\cconfb {S} {\avar x} &\to \sconfb {S} {\tvar x}\label{tr:10}\\
\cconfb {S} {\tapp {\val_1} {\val_2}} &\to \cconfb {\cons {\flapp {\val_1}} S} {\val_2}\label{tr:11}\\
\cconfb {S} {\cache \ell \val} &\to \mconfb {\lookup H \ell} \ell {S} \val\label{tr:12}\\
\mconfb {\some \term} \ell {S} \val &\to \sconfb {S} \term\label{tr:13}\\
\mconfb \none \ell {S} \val &\to \cconfb {\cons {\fcache \ell} {S}} \val\label{tr:14}\\
\sconfb {\cons {\fcache \ell} {S}} \term &\to \update {\sconfb {S} \term} \ell {\some \term}\label{tr:15}\\
\sconfb {\cons {\flapp \val} {S}} \term &\to \cconfb {\cons {\frapp \term} {S}} \val\label{tr:16}\\
\sconfb {\cons{\frapp {\term_2}} {S}} {\term_1} &\to \sconfb {S} {\tapp {\term_1} {\term_2}} \label{tr:17}\\
\sconfb {\cons {\flam x} {S}} \term  &\to \sconfb {S} {\tlam x \term}\label{tr:18}\\
\sconfb \nil \term &\mapsto \term \nonumber
\end{align}
\caption{Transitions of the substitution-based machine}\label{sKNCV}
\end{figure}

\section{A substitution-based higher-order evaluator}
\label{app:evaluator}

Using functional correspondence between evaluators and abstract
machines it is also possible to derive a higher-order evaluator
corresponding to the substitution-based machine.  Here we present a
fragment of this evaluator.  The first pattern matching reveals that
the type of terms is extended with the constructor for substitution
delimiters.

\begin{lstlisting}
let env_lookup : term -> value =
  function
  | Var x        -> abstract_variable (x ^ "_free")
  | Subs v       -> v
  | _            -> assert false

let rec eval (t : term) : value =
  match t with
  | Lam (x, t')  -> Abs (x,
      fun v -> eval (subst x (Subs (mount_cache v)) t'))
  | App (t1, t2) -> let v2 = eval t2
                    in from_sem (eval t1) v2
  | x            -> env_lookup x
\end{lstlisting}

\section{Soundness of the machine}
\label{sec:soundness}

The soundness of the obtained machines with respect to the rrCbV
strategy could be argued by reasoning starting from the soundness of
KNV~\cite{BiernackaBCD20} and showing its preservation by the
subsequent program transformations.
Instead, we present a sketch of a direct proof.

\subsection{Decoding of the machine}\label{ssec:decoding}

Below we define a decoding of stacks to contexts and
of machine terms, values and configurations to source terms.

\begin{equation*}
\begin{aligned}[c]
\semt{\tapp {\term_1} {\term_2} } &= \tapp {\semt{\term_1}} {\semt{\term_2}}\\
\semt{\tlam x \term}              &= \tlam x {\semt{\term}}\\
\semt{\tvar x}                    &= \tvar x\\
\semt{\delimit \val}              &= \semv{\val}\\
\semv{\tapp {\val_1} {\val_2} }   &= \tapp {\semv{\val_1}} {\semv{\val_2}}\\
\semv{\abstrb x \term}            &= \tlam x {\semv{\term}}\\
\semv{\avar x}                    &= \tvar x\\
\semv{\cache \ell \val     }      &= \semv{\val}
\end{aligned}
\begin{aligned}[c]
\sems{\nil}                      &= \hole\\
\sems{\cons {\flapp \term} S}    &= \plug {\sems S} {\flapp {\semt \term}}\\
\sems{\cons {\frapp \val} S}     &= \plug {\sems S} {\frapp {\semv \val}}\\
\sems{\cons {\flapp \val} S}     &= \plug {\sems S} {\flapp {\semv \val}}\\
\sems{\cons {\frapp \term} S}    &= \plug {\sems S} {\frapp {\semt \term}}\\
\hspace{1cm}
\sems{\cons {\flam x} S}         &= \plug {\sems S} {\flam x}\\
\sems{\cons {\fcache \ell} S}    &= \sems S\\
~
\end{aligned}
\end{equation*}
\begin{equation*}
\begin{aligned}[c]
\semk{ \econfb \term S } &= \plug {\sems S} {\semt \term} \\
\semk{ \cconfb S \val }    &= \plug {\sems S} {\semv \val}\\
\semk{ \mconfb {\optional \term} \ell S \val }  &= \plug {\sems S} {\semv \val}
\hspace{1cm}\\
\semk{ \sconfb 	S \term}  &=  \plug {\sems S} {\semt \term}
\end{aligned}
\end{equation*}

\subsection{Shape invariants}\label{ssec:shape}

We state more precise shape invariants to assert that evaluation
contexts and values follow the rules of the rrCbV strategy.  Such
invariants can be derived from evaluators explicitly expressing their
own invariants.  In the grammars below numerical subscripts will also
discriminate grammar symbols. The syntactic categories $n$ and
$a$ of normal forms and neutral terms are those defined in
Section~\ref{sec:cbv}.
\begin{alignat*}{4}
{{\val}_\wnf} &::= {\abstrb x \term} &\alt& \val_\inert &\alt& \cache \ell {\val_\wnf} \hspace{2cm}&
\optional \nf & ::= \none \alt \some \nf\\
\val_\inert &::= {\avar x} &\alt& \tapp {\val_\inert} {\val_\wnf} &\alt& \cache \ell {\val_\inert} &
\optional \neu & ::= \none \alt \some \neu
\end{alignat*}
\begin{alignat*}{11}
S_1 &::=&
       \cons{\flapp \term}{S_1}
  &\alt& \cons{\frapp {\val_\wnf}}{S_1}
  &\alt {S_3}\\
S_2 &::=&
		\cons {\fcache \ell} {S_2}
	&\alt& \cons{\frapp{\nf}}{S_2}
	&\alt S_3 \\
S_3 &::=\;&
		\cons {\fcache \ell} {S_3}
	&\alt& \cons{\flapp{\val_\inert}}{S_2}
	&\alt \cons{\flam x}{S_3}
	\alt \nil
\end{alignat*}

\begin{lemma}\label{lem:shape}
  All reachable configurations are well-formed, \ie, are in forms:
  $\econfb \term {S_1}$,
  $\cconfb {S_1} {\val_\wnf}$,
  $\cconfb {S_2} {\val_\inert}$, 
  $\mconfb {\optional \nf}  \ell {S_3} {\val_\wnf}$,
  $\mconfb {\optional \neu} \ell {S_2} {\val_\inert}$,
  $\sconfb {S_2} \neu$,
  $\sconfb {S_3} \nf$.
\end{lemma}
\begin{proof}[idea]The initial configuration is well-formed and is preserved by all transitions.\end{proof}

\begin{corollary}\label{cor:Rctx}
  If $S$ is a reachable stack of the machine then context $\sems S$ is a rrCbV context.
\end{corollary}

\begin{proof}[sketch]
  In~\cite{BiernackaBCD20} it is shown that all rrCbV contexts are
  generated by the (outside-in) grammar of contexts from
  Section~\ref{sec:cbv}, with the starting symbol $R$. It is also shown
  that this grammar is equivalent to the following (inside-out)
  grammar with the starting symbol $S_1$:
\begin{align*}
S_1 &::=
       \plug {S_1} {\flapp \term}
  \alt \plug {S_1} {\frapp \wnf}
  \alt {S_3}\\
S_2 &::=
	 \plug {S_2} {\frapp \nf}
	\alt S_3 \\
  S_3 &::=
	\plug {S_2} {\flapp \inert}
	\alt \plug {S_3} {\flam x}
	\alt \hole
\end{align*}      
Decodings of well-formed stacks follow this grammar. 
\end{proof}

\begin{corollary}\label{cor:wnf}
  If $v$ is a reachable value of the machine
  then term $\semv v$ is a weak normal form.
\end{corollary}

The second corollary states one of the invariants mentioned in
Subsection~\ref{ssec:eam}. To capture the second fact that annotations
$\cache \ell {}$ cannot be stacked, a more precise shape invariant
could be established. This, however, would require more grammar symbols
and well-formed configurations.

\subsection{Interpretation of transitions}

To omit some technical details we focus on the machine soundness for
closed input terms. It is sufficient because open terms can be closed
by abstractions before processing.

\begin{lemma}\label{lem:overhead}
  If $K \stackrel{(\iota)}{\to} K'$,
  $\iota \notin \{\ref{tr:5}, \ref{tr:9}, \ref{tr:13} \}$
  and term $\semk K$ is closed
  then $\semk K = \semk {K'}$.
\end{lemma}
\begin{proof} By case analysis on transition rules.\end{proof}

\begin{lemma}\label{lem:beta}
  If $K$ is a reachable configuration,
  term $\semk K$ is closed
  and $K \stackrel{(\ref{tr:5})}{\to} K'$
  then $\semk K \red R {\beta_\wnf} \semk {K'}$.
\end{lemma}
\begin{proof}[sketch]
  Transitions maintain the invariant that all free variables of terms
  under delimiters are bound by the stack.  Hence they are not
  captured during substitution, and substitution can be delimited:
  the substituted variable does not occur under delimiter and
  $\beta$-contraction is simulated properly.  From
  Corollaries~\ref{cor:Rctx} and~\ref{cor:wnf} it follows that an
  evaluation context is an $R$-context and a substituted value decodes
  to a weak normal form.
\end{proof}

\begin{lemma}\label{lem:alpha}
  If $K$ is a reachable configuration,
  term $\semk K$ is closed
  and $K \stackrel{(\ref{tr:9})}{\to} K'$
  then $\semk K \red {} {\alpha} \semk {K'}$.
\end{lemma}

\begin{proof}[sketch]
  Thanks to the fresh variable $\alpha$-contraction is simulated correctly.
  As in Lemma~\ref{lem:beta}, the substituted variable (here $x$) does
  not occur under delimiters and the free variable ${\fresh x}$ is
  bound by the stack.

\end{proof}

\begin{lemma}[bypass]\label{lem:bypass}
  If $K$ is a reachable configuration,
  term $\semk K$ is closed
  and $K \stackrel{(\ref{tr:13})}{\to} K'$
  then $\semk K \rred R {\beta_\wnf} \converts {} \alpha \semk {K'}$.
\end{lemma}
\begin{proof}[idea]
  A normal form can be memoized only by getting off a $\fcache \ell$
  frame by transition (\ref{tr:15}).  After pushing it on the stack by
  transition (\ref{tr:14}) the only way to do that is to compute the
  full normal form of a given weak normal form.  Thus, if machine had used
  transition (\ref{tr:14}) instead of (\ref{tr:13}) it would have maintained
  the shape invariant, the evaluation context would be still a
  $R$-context, and the computed full normal form would be
  $\alpha$-equivalent.  By standard properties of $\alpha$-conversion
  its uses in transition (\ref{tr:9}) can be postponed.
\end{proof}


\begin{proposition}\label{prop:soundness}
  If $K$ is a reachable configuration,
  term $\semk K$ is closed
  and $K {\rred {} {}} K'$
  then $\semk K \rred R {\beta_\wnf} \converts {} \alpha \semk {K'}$.
\end{proposition}
\begin{proof}
  This is an immediate consequence of Lemmas~\ref{lem:overhead}--\ref{lem:bypass}
\end{proof}

\begin{theorem}[soundness]\label{thm:soundness}
If machine starting from $\term_0$ computes $\term$\\ $(\ie, \econfb {\term_0} \nil {\rred {} {}}
\sconfb \nil {\term})$,
then $\term_0$ reduces in many steps to a normal form $\term$\\
$(\text{\ie}, \text{there } \text{exists } \term' \text{such that } \term_0 \;\rred R {\beta_\wnf}\; \term'$ and $\term' \; \nred R {\beta_\wnf}$ 
and $\term \converts {} \alpha \term')$.
\end{theorem}

\begin{proof}
  By Lemma~\ref{lem:shape} terminal configuration decodes to a term in
  normal form, so $\term$ is in normal form. By
  Proposition~\ref{prop:soundness} $\term_0$ reduces to $\term$.
\end{proof}


\section{Complexity Analysis and Completeness of the Machine}
\label{sec:complexity}
\subsection{The conversion problem}

One of the motivations for constructing an efficient machine for strong
CbV comes from the conversion problem, which asks
if two given terms are $\beta$-convertible. In general the problem is
undecidable, but it has important applications in proof assistants and
thus it is desirable to find efficient partial solutions.

If both input terms normalize in the strong CbV strategy, then
one straightforward solution is to normalize input terms by the
abstract machine and then to check if they are $\alpha$-equivalent.
Normalization yields a shared representation of terms, avoiding
possibly exponential size of the normal forms in explicit
representation.  As Condoluci \etal~state in \cite{CondoluciAC19}, the
$\alpha$-equivalence of shared terms can be checked in time linear
\wrt the size of shared representations.  Therefore convertibility
check can be done in time proportional to the time of normalization of
both terms.  In the following we analyse the cost of normalization.

In~\cite{BiernackaBCD20} it is shown that, using a technique called
\emph{streaming of terms}, the convertibility check can be short-circuited
if partial results of normalization differ. We do not consider it here
as fusion of term streaming and shared equality goes beyond the scope
of this paper.

\subsection{Execution length}
Probably the simplest approach to the complexity analysis of an
abstract machine is to find a (constant) upper bound on the number of
consecutive administrative steps of the machine. Then the overall
complexity of an execution is this bound times the number of
$\beta$-transitions times the cost of a~single transition.

In this section we estimate the global number of transitions executed
by the machine on a given term. Unfortunately, the simple approach
outlined above does not work here as the following example shows. Let
$c_n$ be the $n^\text{th}$ Church numeral,
$\dubleton := \tlam x {\tlam p {\tapp {\tapp {\tvar p} {\tvar x}}
    {\tvar x}}}$ and $I$ -- the identity.  The execution of
$\tapp{\tapp{c_n}{\dubleton}}I$ starts with two $\beta$-reductions
substituting $\dubleton$ and $I$.  Then next $n$ reductions result in
a value $r_n$, where $r_0 = \abstrb x x$ and
$r_{n+} = \abstrb p {\tapp {\tapp {\tvar p} {\delimit {\cache \ell
        {r_n}}}} {\delimit {\cache \ell {r_n}}}}$ for some $\ell$s.
Value $r_n$ decodes to a normal form, but it takes the machine
$\Omega(n)$ administrative steps to construct this normal form. Thus,
sequences of administrative transitions are not bounded by any
constant (see also Fig.~\ref{plt:explode4}).

To overcome the problem of long sequences of administrative steps we
perform a kind of amortized analysis of the execution
length. Following~\cite{Okasaki:99}, we define a potential function
$\psik$ of configurations. Then, for most of the transitions (more
precisely, for all but $\stackrel{(\ref{tr:7})}{\to}$) the cost of the
transition is covered by the change in the potential of the involved
configurations (cf. Lemma~\ref{lem:decrease}). Transition
$\stackrel{(\ref{tr:7})}{\to}$ is a preparatory step for a
$\beta$-reduction and its cost is covered by the involved
$\beta$-reduction (cf. Lemma~\ref{lem:increase}).

The potential of a configuration depends on the potentials of its
components: term, value, stack and the implicit heap. The potential
functions $\psit$, $\psiw$ and $\psis$ are defined as follows.

\begin{equation*}
\begin{aligned}[c]
\Psit{\tapp {\term_1} {\term_2} } &= 6 + \Psit{\term_1} + \Psit{\term_2}\\
\Psit{\tlam x \term}              &= 4 + \Psit{\term}\\
\Psit{\tvar x}                    &= 4\\
\Psit{\delimit \val}              &= 4\\
\Psiw{\tapp {\val_1} {\val_2} }   &= 3 + \Psiw{\val_1} + \Psiw{\val_2}\\
\Psiw{\abstrb x \term}            &= 3 + \Psit{\term}\\
\Psiw{\avar x}                    &= 1\\
\Psiw{\cache \ell \val     }      &= 3 \; \cancel{ + \;\Psiw{\val} }
\end{aligned}
\begin{aligned}[c]
\Psis{\nil}                      &= 0\\
\Psis{\cons {\flapp \term} S}    &= 5 + \Psis S + \Psit \term\\
\Psis{\cons {\frapp \val} S}     &= 4 + \Psis S + \Psiw \val\\
\Psis{\cons {\flapp \val} S}     &= 2 + \Psis S + \Psiw \val\\
\Psis{\cons {\frapp \term} S}    &= 1 + \Psis S\\
\hspace{1cm}
\Psis{\cons {\flam x} S}         &= 1 + \Psis S\\
\Psis{\cons {\fcache \ell} S}    &= 1 + \Psis S\\
~
\end{aligned}
\end{equation*}
Intuitively, these potential functions indicate for how many machine
steps a given construct is responsible. For example,
$\Psit{\tapp {\term_1} {\term_2} }$ says that if an application
$\tapp {\term_1} {\term_2}$ appears somewhere in a configuration, it
may generate 6 transitions of the machine plus the work generated by
the two subterms. The crucial observation is that when a normal form
of a value $\val$ is known and stored under location $\ell$, then
the normalization of $\val$ involves only a constant (more precisely, 2)
steps and it does not involve recomputation of $\val$ -- thus the
amount of work generated by $\cache \ell \val$ is bounded by $3$ (one
transition involves memoizing the normal form).

The potential function $\psih$ estimates the cost of maintaining the
heap (which is implicitly present in each configuration). It takes
into account all values that have their place on the heap (expressed
by the condition ${\cache \ell \val} \in K$ below, meaning that
${\cache \ell \val}$ appears somewhere in the current configuration),
but are not yet normalized (expressed by ${\lookup H \ell} = \none$)
and are currently not being evaluated (expressed by
$\fcache \ell \notin S$, meaning that $\fcache \ell$ does not appear
in the stack and thus the evaluation of $\val$ has not yet
started). It also takes into account the moment when the evaluation of
${\cache \ell \val}$ starts, \ie, when the current configuration is
of the form ${\mconfb \nil \ell S \val}$. Formally, $\psih$ is defined
as follows:

$$\Psih K = \sum_{(\ell, \val) \text{\;\;s.t.\;\;} K = {\mconfb \nil \ell S \val} \;\vee\; ({\cache \ell \val} \in K \;\wedge\; {\lookup H \ell} = \none \;\wedge\; \fcache \ell \notin S)} \Psiw \val$$

Now we define the potential function $\psik$ for configurations.  We use
Iverson bracket in the clause for $\cal C$-configuration ($[\varphi] = 1$
if $\varphi$ and $[\varphi] = 0$ otherwise) to denote advancement of
transition (\ref{tr:6}):
\begin{alignat*}{5}
\Psik{ \econfb \term S } &=& \Psit \term \;+&\; \Psis S + \Psih K\\
\Psik{ \cconfb S \val }    &=&\; \Psiw \val \;+&\; \Psis S + \Psih K - 9 \cdot [\cconfb S \val \stackrel{(\ref{tr:5})}{\to}]\\
\Psik{ \mconfb {\optional \term} \ell S \val }  &=& 2 \;+&\; \Psis S + \Psih K\\
\Psik{ \sconfb S \term}  &=&&\;  \Psis S + \Psih K
\end{alignat*}

\begin{lemma}
\label{lem:delimitann}
If substitution delimiter $\delimit \val$ occurs somewhere
in a reachable configuration
then it is of the form $\delimit {\cache \ell \val}$.
\end{lemma}

\begin{proof}
The only transitions introducing substitution delimiters are
(\ref{tr:5}) and (\ref{tr:9})
which ensure location annotation.
\end{proof}

\begin{lemma}
\label{lem:substpot}
For any $\term$, $x$, $\val$ we have
$\Psit \term = \Psit {\subst x {\delimit \val} \term}$.
\end{lemma}

\begin{proof}
The only constructors that can be replaced by a substitution
are variables and $\Psit x = \Psit {\delimit \val}$.
\end{proof}

\begin{lemma}[decrease]
\label{lem:decrease}
If $K$ is a reachable configuration and
$K \stackrel{\neq(\ref{tr:7})}{\to} K'$ then $ \Psik K > \Psik {K'}$.
\end{lemma}

\begin{proof}[of Lemma \ref{lem:decrease}]
  The proof goes by case analysis on machine transitions.
  \begin{itemize}
  \item[(\ref{tr:1})] ~\\[-2.2em]
    \begin{equation*}
    \begin{split}
      \Psik{ \econfb {\tapp{\term_1}{\term_2}} {S}} &= 6
       + \Psik{\term_1} +\Psik{\term_2} + \Psis S + \Psih K\\
      &> 5
       + \Psik{\term_1} +\Psik{\term_2} + \Psis S + \Psih K\\
        & =\Psik{\econfb {\term_2} {\cons {\flapp{\term_1}} {S}}}
      \end{split}
    \end{equation*}
  \item[(\ref{tr:2})] ~\\[-2.2em]
    \begin{equation*}
    \begin{split}
      \Psik{\econfb {\tlam x \term} {S}} &= 4
       + \Psik{\term} + \Psis S + \Psih K\\
       &> 3
       + \Psik{\term} + \Psis S + \Psih K -9 \cdot [\cconfb S {\abstrb x \term} \stackrel{(\ref{tr:5})}{\to}]\\
      &= \Psik{\cconfb {S} {\abstrb x \term}}\\
      \end{split}
    \end{equation*}
 \item[(\ref{tr:3})] Here, by Lemma~\ref{lem:delimitann}, term $\term$ is of the form either
   $\tvar x$ or $\delimit {\cache \ell \val}$, so $\strip \term$ is either $\avar x$ or
   $\cache \ell \val$ and thus
    \begin{equation*}
      \begin{split}
        \Psik{\econfb {\term} {S}} &=  4+ \Psis S + \Psih K\\
       &> 3 + \Psis S + \Psih K\\
       &\geq \Psiw {\strip \term} + \Psis S + \Psih K - 9 \cdot [\cconfb S \val \stackrel{(\ref{tr:5})}{\to}]\\
       &= \Psik{\cconfb {S} {\strip \term}}
    \end{split}
  \end{equation*}
 \item[(\ref{tr:4})] ~\\[-2.2em]
    \begin{equation*}
      \begin{split}
        \Psik{\cconfb {\cons {\flapp \term} {S}} \val} &=  \Psiw{\val}
        + 5+ \Psis S + \Psit \term + \Psih K - 0\\
&> \Psit \term + 4 + \Psis S + \Psiw \val + \Psih K\\
&= \Psik{\econfb \term {\cons{ \frapp \val} {S}}}
    \end{split}
  \end{equation*}
 \item[(\ref{tr:5})] ~\\[-2.2em]
    \begin{equation*}
    \begin{split}\Psik {\cconfb {\cons {\frapp {\cache \ell \val}\!\!}
          {\!S}} {\abstrb x \term}} &=
      \Psiw{\abstrb x \term} +\Psis{\cons {\frapp {\cache \ell \val}\!\!}
        {\!S}} + \Psih K -9 \\
      &= 3 + \Psit \term + 4 + \Psis S + 3 + \Psih K -9 \\
      &=  \Psit \term  +\Psis S  + \Psih K +1\\
      &\stackrel{\text{Lemma}~\ref{lem:substpot}}{>} \Psit {\subst x {\delimit {\cache \ell \val}} \term}  +\Psis S  + \Psih K\\
      &=\Psik{\econfb {\subst x {\delimit {\cache \ell \val}} \term} {S}}
    \end{split}
  \end{equation*}
 \item[(\ref{tr:6})] ~\\[-2.2em]
    \begin{equation*}
    \begin{split}\Psik{\cconfb {\cons {\frapp  {\cache {\phantom \ell}
              \val}\!\!} {\!S}} {\abstrb x \term}}
      &= \Psiw{\abstrb x \term} +\Psis{{\cons {\frapp  {\cache {\phantom \ell}
              \val}\!\!} {\!S}}} + \Psih K - 0\\
      &= \Psiw{\abstrb x \term} + 4 + \Psis{S} + \Psiw{\val} + \Psih K - 0\\
      &> \Psiw{\abstrb x \term} + 4 + \Psis S  + 3 + \Psiw{\val} + \Psih K - 9\\
      &= \Psiw{\abstrb x \term} + 4 + \Psis S  + \Psiw{\cache \ell \val} + \Psih {K'} 
       - 9\\
      &=\Psik{\cconfb {\cons {\frapp {\cache \ell \val}} {S}} {\abstrb x \term}}
    \end{split}
  \end{equation*}
\item[(\ref{tr:8})] This is an easy case: $4$ occurring in $\Psis{\cons {\frapp \val_2} S}$ is greater than $3$ occurring in $\Psiw{\tapp {\val_1} {\val_2} }$.
 \item[(\ref{tr:9})] ~\\[-2.2em]
    \begin{equation*}
    \begin{split}\Psik {\cconfb S {\abstrb x \term}} &=
      \Psiw{\abstrb x \term} +\Psis S + \Psih K - 0 \\
      &= 3 + \Psit \term + \Psis S + \Psih K \\
      &> \Psit \term + 1 + \Psis S + \Psiw {\avar {\fresh x}} + \Psih K\\
      &\stackrel{\text{Lemma}~\ref{lem:substpot}}{=}
       \Psit {\subst x {\delimit {\cache \ell {\avar {\fresh x}}}} \term}  + \Psis {\cons {\flam {\fresh x}} {S}}  + \Psih {K'}\\
      &=\Psik{\econfb {\subst x {\delimit {\cache \ell {\avar {\fresh x}}}} \term} {\cons {\flam {\fresh x}} {S}}}
    \end{split}
  \end{equation*}
 \item[(\ref{tr:10})] This is an easy case.
 \item[(\ref{tr:11})] This is an easy case.
 \item[(\ref{tr:12})] ~\\[-2.2em]
    \begin{equation*}
    \begin{split}\Psik{\cconfb {S} {\cache \ell \val}} &=
    3 + \Psis S + \Psih K\\
    &> 2 + \Psis S + \Psih K\\
    &= \Psik{ \mconfb {\lookup H \ell} \ell {S} \val}
    \end{split}
  \end{equation*}
 \item[(\ref{tr:13})] ~\\[-2.2em]
    \begin{equation*}
    \begin{split}\Psik{\mconfb {\some \term} \ell {S} \val} &=
    2 + \Psis S + \Psih K\\
    &> \Psis S + \Psih K\\
    &= \Psik{\sconfb {S} \term}
    \end{split}
  \end{equation*}
 \item[(\ref{tr:14})] ~\\[-2.2em]
    \begin{equation*}
    \begin{split}\Psik{\mconfb \none \ell {S} \val} &=
    2 + \Psis S + \Psih K\\
    &> 1 + \Psis S + \Psih K\\
    &= \Psiw \val + 1 + \Psis S + \Psih {K'}\\
    &= \Psik{\cconfb {\cons {\fcache \ell} {S}} \val}
  \end{split}
  \end{equation*}
 \item[(\ref{tr:15})] The pair $(\ell, \val)$ is not counted before the transition 
    because ${\fcache \ell}$ is on the stack, and it is not counted after the transition
    because $\lookup H \ell \neq \none$.
    \begin{equation*}
    \begin{split}\Psik{\sconfb {\cons {\fcache \ell} {S}} \term} &=
    1 + \Psis S + \Psih K\\
    &> \Psis S + \Psih K\\
    &= \Psik{{\sconfb {S} \term}}
    \end{split}
  \end{equation*}
\item[(\ref{tr:16})] This is an easy case: $2$ occurring in
  $\Psis{\cons {\flapp \val} S} $ is greater than $1$ occurring in
  $\Psis{\cons {\frapp \term} S}$.  The value counted on the stack
  before the transition is counted in the configuration after the
  transition.
 \item[(\ref{tr:17})] This is an easy case.
 \item[(\ref{tr:18})] This is an easy case.
\end{itemize}
\end{proof}

\begin{lemma}[subterm]
\label{lem:subterm}
If $\abstrb x t$ is a reachable value of the machine
starting from term $t_0$
then $\Psiw {\abstrb x t} < \Psit {t_0}$.
\end{lemma}

\begin{proof}
Both machines never perform reductions
in bodies of abstractions which will be
invoked later.
In environment-based machine for all values $\abstr x t E$
terms $\tlam x t$ are always subterms of $t_0$.
By translation in substitution-based machine
these abstractions may be modified
only by replacing source variables
with values under substitution delimiters
and by Lemma~\ref{lem:substpot} it does not change the potential.
Function $\psit$ assigns 4 to abstraction constructor
which is greater than 3 assigned by $\psiw$.

\end{proof}

\begin{lemma}[increase]
\label{lem:increase}
  If $K$ is a reachable configuration from $\term_0$ and
  $K \stackrel{(\ref{tr:7})}{\to} K'$  then
  $ \Psik {K'} - \Psik K < \Psit {\term_0} $.
\end{lemma}

\begin{proof}
~\\\noindent  
$\Psik{\cconfb {\cons {\frapp  
              \val} {\!S}} {\cache \ell {\abstrb x \term}}} + \Psit {\term_0}$\vspace{-2ex}
  \begin{align*}
\qquad      &= \Psiw{{\cache \ell {\abstrb x \term}}} +\Psis{{\cons {\frapp  \val} {\!S}}} + \Psih K - 0 + \Psit {\term_0}\\
      &= \Psit {\term_0} + 3 + \Psis{{\cons {\frapp
              \val} {\!S}}} + \Psih K - 0\\
      &\stackrel{\text{Lemma}~\ref{lem:subterm}}{>}
       \Psiw{\abstrb x \term} + \Psis{{\cons {\frapp  
              \val} {\!S}}} + \Psih {K'} - 9 \cdot [\cconfb {\cons {\frapp
              \val} {\!S}} {\abstrb x \term} \stackrel{(\ref{tr:5})}{\to}]\\
      &=\Psik{\cconfb {\cons {\frapp \val} {\!S}} {\abstrb x \term}}
    \end{align*}
\end{proof} 
Example changes of potential are depicted in Figures~\ref{plt:explode4} and \ref{plt:test1} using
Matplotlib~\cite{Hunter:2007}.

\begin{figure}[h]
\begin{center}
\includegraphics[scale=.75]{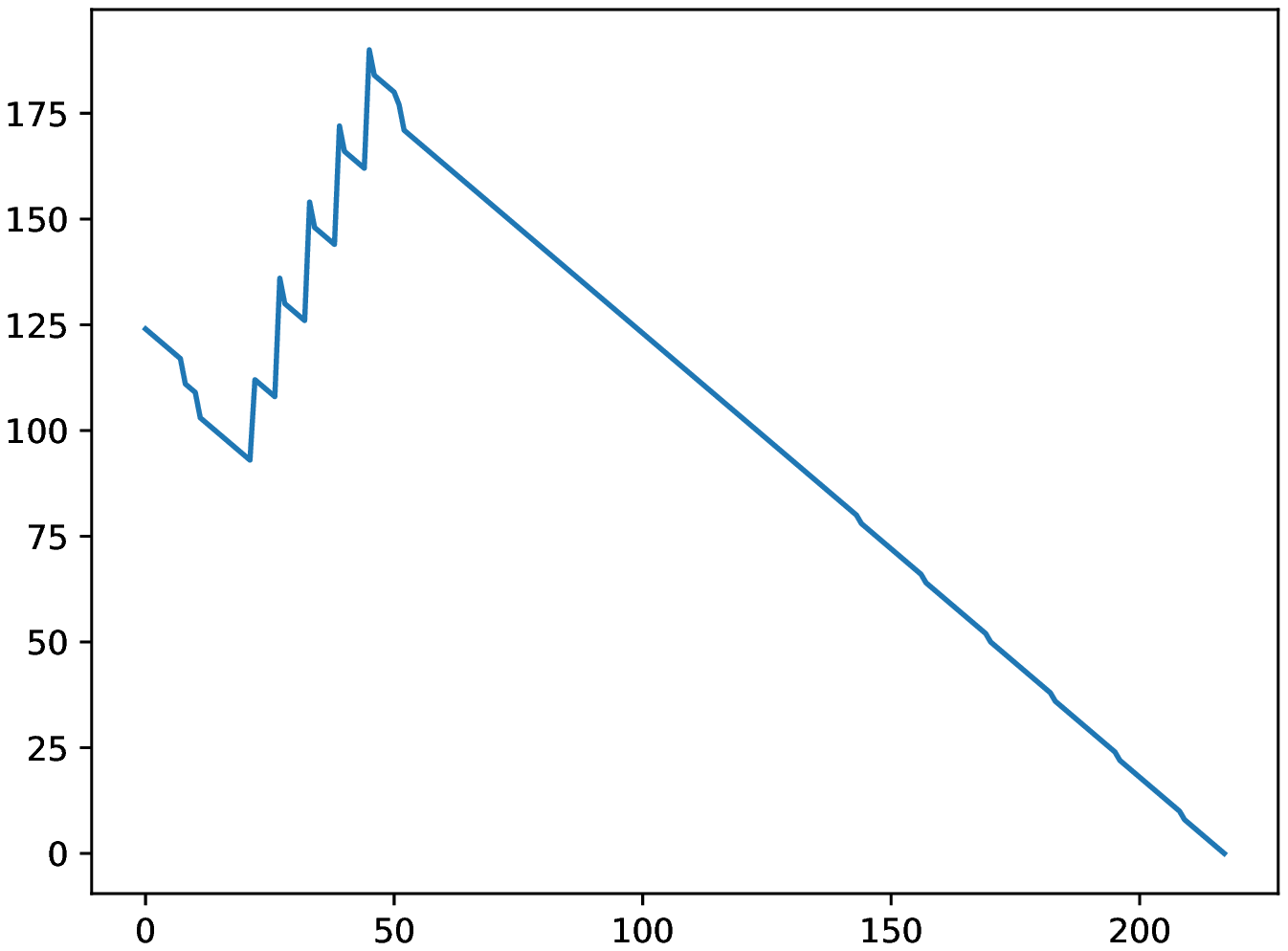}

\caption{Plot of potential for execution of
$\tapp{\tapp{c_6}{\dubleton}} I$
} performing 217 transitions of which 8 are $\beta$-transitions (\ref{tr:5}) \label{plt:explode4}
\end{center}
\end{figure}

\begin{figure}[h]
\begin{center}
\includegraphics[scale=.75]{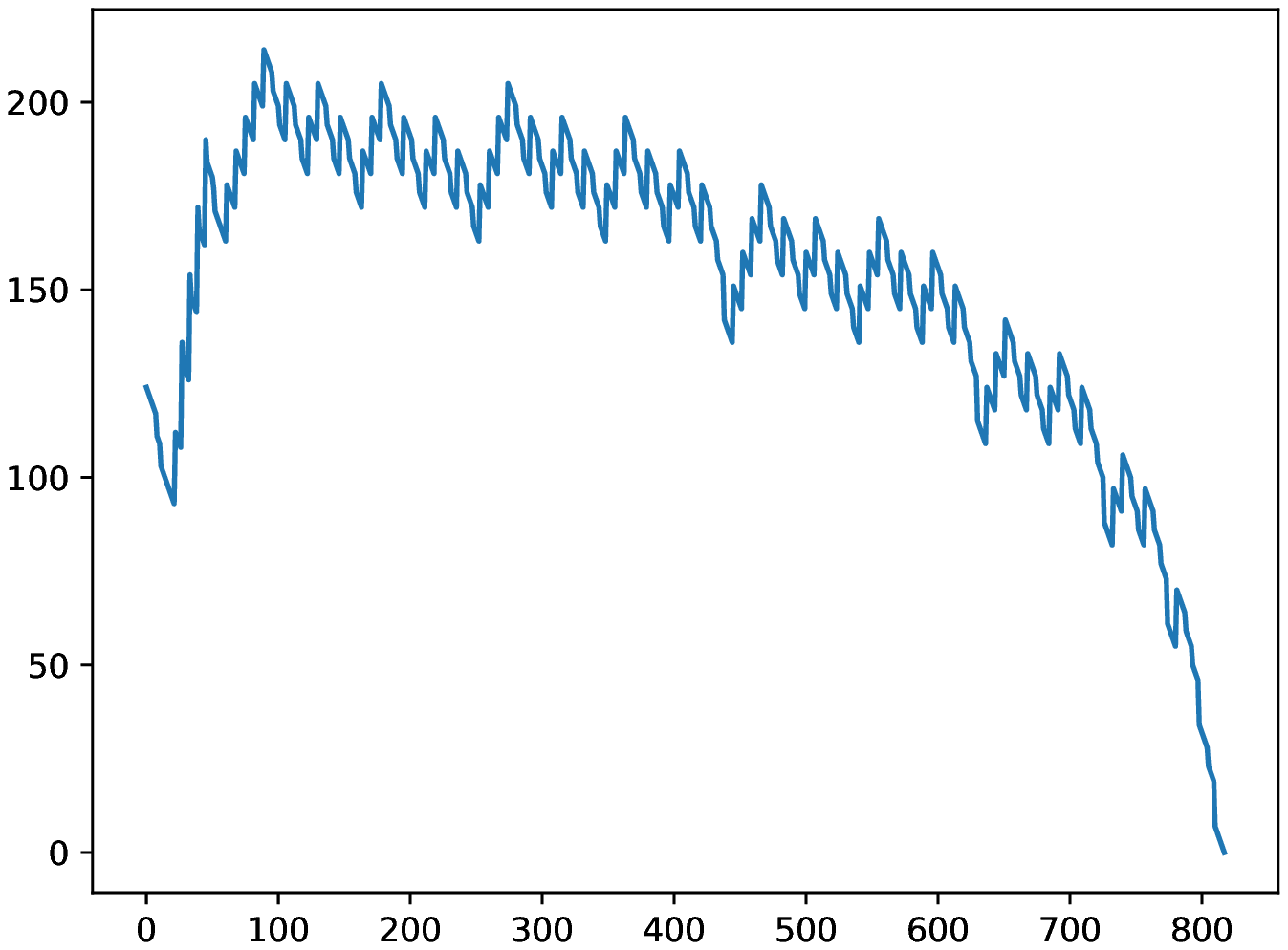}

\caption{Plot of potential for execution of
$\tapp{\tapp{c_6}{c_2}} I$
} performing 817 transitions of which 134 are $\beta$-transitions (\ref{tr:5}) \label{plt:test1}
\end{center}
\end{figure}

\begin{lemma}
\label{lem:tracebound}
  Let $\rho$ be a sequence of consecutive machine transitions starting from term $\term_0$, $|\rho|$ be number of steps in $\rho$ and  $|\rho|_{(\ref{tr:7})}$ be number of steps $(\ref{tr:7})$ in $\rho$. Then
  $|\rho| \leq (|\rho|_{(\ref{tr:7})} +1) \cdot \Psit {\term_0}$.
\end{lemma}
\begin{proof} This is an immediate consequence of  Lemmas \ref{lem:decrease} and \ref{lem:increase}.
\end{proof} 

\subsection{Completeness}

The upper bound on a length of machine trace from Lemma
\ref{lem:tracebound} leads to the completeness of the machine.

\begin{proposition}\label{prop:completeness}
If $K$ is a reachable configuration,
term $\semk K$ is closed
and\\ $\semk K \; \red R {\beta_\wnf} \; t'$,
then there exists $K'$ such that $K {\rred {} {}} K'$
and $t' \; \rred R {\beta_\wnf}\! \converts {} \alpha \; \semk {K'} $.
\end{proposition}

\begin{proof}
  By Lemma~\ref{lem:shape} and inspection of all transitions,
  configurations that decode to a term not in normal form are
  non-terminal.  Since the decoding of the given configuration
  $\beta$-reduces, by Proposition~\ref{prop:soundness} the machine
  makes steps until $(\ref{tr:5})$ or $(\ref{tr:13})$ is performed.
  By Lemma~\ref{lem:decrease} the number of consecutive transitions
  $\neq\!\!(\ref{tr:7})$ is bounded by the potential. Since transition
  $(\ref{tr:7})$ must be followed by $(\ref{tr:5})$ in one or two
  steps, $(\ref{tr:5})$ or $(\ref{tr:13})$ must be reached
  eventually. In both cases Proposition~\ref{prop:soundness}
  guarantees that, up to $\alpha$-conversion, the same term is reached
  or bypassed.
\end{proof}

\begin{theorem}[completeness]
  If $\term_0$ reduces in many steps to a normal form $\term$
  $(\text{\ie}, \term_0 \;\rred R {\beta_\wnf}\; \term$ and
  $\term \; \nred R {\beta_\wnf})$, then machine starting from
  $\term_0$ computes $\term$
  $(\ie, \text{there }\\ \text{exists } \term' \text{such that }
  \econfb {\term_0} \nil {\rred {} {}} \sconfb \nil {\term'}$ and
  $\term \converts {} \alpha \term')$. $\phantom{\rred R{}}$
\end{theorem}

\begin{proof}
  By Lemma \ref{lem:tracebound} the machine reaches a terminal
  configuration in finite number of steps. By Lemma~\ref{lem:shape}
  the terminal configuration decodes to a term in normal form and by
  Proposition~\ref{prop:completeness} this normal form is
  $\alpha$-equivalent to $\term$.
\end{proof}

\subsection{Implementation of environments}
\label{ssec:transition-costs}

The environment-based machine performs insert (in rules (\ref{tre:5})
and (\ref{tre:9})) and lookup (in rule (\ref{tre:3})) operations on
environments. When environments are implemented as lists, then the
cost of lookup is proportional to the size of the list.

Every environment is paired with a subterm of the initial term
constituting a closure.  The machine maintains an invariant that the
size of any environment is equal to the number of lambda abstractions
under which the paired subterm is located in the original term.

Accattoli and Barras in \cite{Accattoli-Barras:PPDP17} outline two
realizations of local environments which improve the asymptotic cost of
operations. One of them uses balanced trees, which are a natural choice
for named representation of lambda expressions. The other uses
random-access lists -- this requires precomputing de Bruijn indices of
variables in the original term.  Both these data structures are
persistent, both improve the cost of lookup to logarithmic in the size
of the initial term.

We also note that if identifiers are strings of symbols from a finite
alphabet (as binary numbers are strings of bits) then tries can be
applied as dictionaries making lookup and update costs proportional to
the length of the involved identifier.

\subsection{Overall complexity}
We treat arithmetic operations and operations on identifiers as realizable in
constant time.  In fact, they involve extra logarithmic cost but it
does not affect the polynomial complexity.  As described in Subsection
\ref{ssec:eam}, heap operations can be implemented within constant time.

Let $n$ be the number of $\beta$-reductions performed in a derivation from
term $\term_0$.  It can be simulated on both machines with
$O( (1+n) \cdot |\term_0|)$ transitions.

The environment-based machine has three transitions ((\ref{tre:3}),
(\ref{tre:5}) and (\ref{tre:9})) whose  cost is related to the environment;
all the other transitions have constant cost.  Therefore the overall cost of
the execution is $O( (1+n) \cdot |\term_0| \cdot {E(|t_0|)})$ where
${E(|t_0|)}$ is the cost of an operation on an environment of size
$|t_0|$ and this cost can be considered as logarithmic
because
an environment maps variables occurring in the input term
to values
and number of such variables is bounded by $|t_0|$.

In the substitution-based machine transitions (\ref{tr:5}) and
(\ref{tr:9}) have cost proportional to $|t_0|$ while the cost of all
the other transitions is constant.  Therefore the overall cost of the
execution is $O( (1+n) \cdot |\term_0|^2)$.  However, in the weak case,
\ie, before going under $\lambda$ with transition (\ref{tr:9}), the
machine performs $O( (1+n) \cdot |\term_0|)$ transitions with unitary
cost and $O(n)$ transitions (\ref{tr:5}) of cost $O(|\term_0|)$.
Therefore the cost of the weak part of the execution is
\textit{bilinear}, \ie, $O( (1+n) \cdot |\term_0|)$.

Accattoli and Dal Lago, in \cite{AccattoliL12,DalLagoA/abs-1711-10078},
present a polynomial simulation of Turing Machines relevant for any strategy reducing weak redexes first
which is also true for rrCbV.
We have shown a polynomial simulation of
rrCbV strategy.  Therefore reasonable
machines and the rrCbV strategy can simulate
each other with a polynomial overhead making the latter also a
reasonable machine for time.

Call-by-value strategies as described in the preliminaries perform the
same number of $\beta$-reductions to achieve normal form so this
result generalizes to all of them:

\begin{theorem}
  The number of steps performed by a strong call-by-value strategy is a
  reasonable measure for time.
\end{theorem}


\section{Conclusion and future work}
\label{sec:conclusion}
We presented an abstract machine that realizes the strong CbV strategy
(in its right-to-left variant) and we proved its reasonability that
makes it a sufficiently good implementation model. The machine uses a
form of memoization to store computed normal forms and reuse them when
needed. A derivation from an evaluator using memothunks also in weak
normalization would lead to a machine performing some kind of strong
call-by-need strategy but its study and comparison with other
call-by-need machines are beyond the scope of this paper.


\bibliography{main}

\begin{thebibliography}{10}
\providecommand{\url}[1]{\texttt{#1}}
\providecommand{\urlprefix}{URL }
\providecommand{\doi}[1]{https://doi.org/#1}

\bibitem{DBLP:conf/rta/Accattoli19}
Accattoli, B.: A fresh look at the lambda-calculus (invited talk). In:
  4th International Conference on Formal Structures for Computation
  and Deduction, {FSCD} 2019. LIPIcs, vol.~131, pp. 1:1--1:20 (2019)

\bibitem{Accattoli-Barras:PPDP17}
Accattoli, B., Barras, B.: Environments and the complexity of abstract
  machines. In: Proceedings of the 19th
  International Symposium on Principles and Practice of Declarative Programming
  (PPDP'17). pp. 4--16 (2017)

\bibitem{Accattoli-Coen:LICS15}
Accattoli, B., Coen, C.S.: On the relative usefulness of fireballs. In: 30th
  Annual {ACM/IEEE} Symposium on Logic in Computer Science, {LICS} 2015.
  pp. 141--155 (2015)

\bibitem{DBLP:conf/ppdp/AccattoliCGC19}
Accattoli, B., Condoluci, A., Guerrieri, G., Coen, C.S.: Crumbling abstract
  machines. In: Proceedings of the 21st International
  Symposium on Principles and Practice of Programming Languages, {PPDP} 2019.
  pp. 4:1--4:15 (2019)

\bibitem{AccattoliL12}
Accattoli, B., {Dal Lago}, U.: On the invariance of the unitary cost model for
  head reduction. In: 23rd International Conference on
  Rewriting Techniques and Applications, {RTA} 2012.
  LIPIcs, vol.~15, pp. 22--37 (2012)

\bibitem{Accattoli-DalLago:LMCS16}
Accattoli, B., Lago, U.D.: ({L}eftmost-outermost) beta reduction is invariant,
  indeed. Logical Methods in Computer Science  \textbf{12} (2016)

\bibitem{AccattoliG16}
Accattoli, B., Guerrieri, G.: Open call-by-value. In: 14th Asian Symposium,
  {APLAS} 2016, Proceedings. LNCS, vol. 10017, pp. 206--226 (2016)

\bibitem{DBLP:journals/scp/AccattoliG19}
Accattoli, B., Guerrieri, G.: Abstract machines for open call-by-value. Sci.
  Comput. Program.  \textbf{184} (2019)

\bibitem{Aehlig-Joachimski:MSCS04}
Aehlig, K., Joachimski, F.: Operational aspects of untyped normalization by
  evaluation. Mathematical Structures in Computer Science  \textbf{14},
  587--611 (2004)

\bibitem{Ager-al:RS-03-14}
Ager, M.S., Biernacki, D., Danvy, O., Midtgaard, J.: From interpreter to
  compiler and virtual machine: a functional derivation. Tech. Rep. BRICS
  RS-03-14, DAIMI, Aarhus University, Aarhus, Denmark (Mar 2003)

\bibitem{Ager-al:PPDP03}
Ager, M.S., Biernacki, D., Danvy, O., Midtgaard, J.: A functional
  correspondence between evaluators and abstract machines. In: Proceedings of
  the Fifth ACM-SIGPLAN Conference, PPDP'03. pp. 8--19 (2003)

\bibitem{Balabonski-al:ICFP17}
Balabonski, T., Barenbaum, P., Bonelli, E., Kesner, D.: Foundations of strong
  call by need. {PACMPL}  \textbf{1}({ICFP}),  20:1--20:29 (2017)

\bibitem{BiernackaBCD20}
Biernacka, M., Biernacki, D., Charatonik, W., Drab, T.: An abstract machine for
  strong call by value. In: Programming Languages
  and Systems - 18th Asian Symposium, {APLAS} 2020, Proceedings.
  LNCS, vol. 12470, pp. 147--166 (2020)

\bibitem{BiernackaC19}
Biernacka, M., Charatonik, W.: Deriving an abstract machine for strong call by
  need. In: 4th International Conference on Formal Structures for Computation
  and Deduction, {FSCD} 2019. LIPIcs, vol.~131, pp. 8:1--8:20 (2019)

\bibitem{BChZ-fscd17}
Biernacka, M., Charatonik, W., Zielinska, K.: Generalized refocusing: From
  hybrid strategies to abstract machines. In: 2nd International Conference on
  Formal Structures for Computation and Deduction, {FSCD} 2017. pp. 10:1--10:17
  (2017)

\bibitem{CondoluciAC19}
Condoluci, A., Accattoli, B., Coen, C.S.: Sharing equality is linear. In:
  Proceedings of the 21st International Symposium on Principles and Practice of
  Programming Languages, {PPDP} 2019. pp. 9:1--9:14 (2019)

\bibitem{Cregut:HOSC07}
Cr\'egut, P.: Strongly reducing variants of the {K}rivine abstract machine.
  Higher-Order and Symbolic Computation  \textbf{20}(3),  209--230 (2007)

\bibitem{DalLagoA/abs-1711-10078}
{Dal Lago}, U., Accattoli, B.: Encoding {T}uring machines into the
  deterministic lambda-calculus. CoRR  \textbf{abs/1711.10078} (2017)

\bibitem{Filinski-Rohde:RAIRO05}
Filinski, A., Rohde, H.K.: Denotational aspects of untyped normalization by
  evaluation. Theoretical Informatics and Applications  \textbf{39}(3),
  423--453 (2005)

\bibitem{Garcia-Nogueira:SCP14}
Garc{\'\i}a-P{\'e}rez, A., Nogueira, P.: On the syntactic and functional
  correspondence between hybrid (or layered) normalisers and abstract machines.
  Science of Computer Programming  \textbf{95},  176--199 (2014)

\bibitem{Gregoire-Leroy:ICFP02}
Gr\'egoire, B., Leroy, X.: A compiled implementation of strong reduction. In:
  International Conference on Functional Programming. pp. 235--246. SIGPLAN
  Notices~37(9) (2002)

\bibitem{Hunter:2007}
Hunter, J.D.: Matplotlib: A 2d graphics environment. Computing in Science \&
  Engineering  \textbf{9}(3),  90--95 (2007)

\bibitem{DBLP:journals/tcs/LagoM08}
Lago, U.D., Martini, S.: The weak lambda calculus as a reasonable machine.
  Theor. Comput. Sci.  \textbf{398}(1-3),  32--50 (2008)

\bibitem{Leroy-al:OCaml-4.10}
Leroy, X., Doligez, D., Frisch, A., Garrigue, J., Rémy, D., Vouillon, J.: The
  {OC}aml system, release 4.10. INRIA, Rocquencourt, France (Feb 2020)

\bibitem{Okasaki:99}
Okasaki, C.: Purely functional data structures. Cambridge University Press
  (1999)

\bibitem{Sestoft:Jones02}
Sestoft, P.: Demonstrating lambda calculus reduction. In: The Essence of
  Computation: Complexity, Analysis, Transformation. Essays Dedicated to Neil
  D. Jones. pp. 420--435. No.~2566 in Lecture Notes in Computer Science (2002)

\bibitem{DBLP:conf/stoc/SlotB84}
Slot, C.F., van Emde~Boas, P.: On tape versus core; an application of space
  efficient perfect hash functions to the invariance of space. In:
  Proceedings of the 16th Annual {ACM} Symposium on Theory of
  Computing. pp. 391--400. {ACM} (1984)

\bibitem{supplement}
\url{https://www.ii.uni.wroc.pl/~tdr/pub/rscbv.zip}

\end{thebibliography}
\bibliographystyle{splncs04}


\end{document}